\documentclass[11pt]{article}

\title{Bounds for Hardness Condensation in the Query Model\footnote{Combined version of arXiv:2602.00754 and
  arXiv:2602.01042. A preliminary version of this paper has been accepted for presentation at 41$^{st}$ Computational Complexity Conference (CCC 2026)}}
\author{Chandrima Kayal\footnote{Universit\'e Paris Cit\'e, CNRS, IRIF, Paris, France, Email:{\tt kayal@irif.fr}}
 \and Rajat Mittal\footnote{Indian Institute of Technology Kanpur. Email{\tt rmittal@cse.iitk.ac.in}}
 \and Sai Soumya Nalli\footnote{Microsoft Research India. Email:{\tt saisoumya7208@gmail.com}}
 \and
 Manaswi Paraashar\footnote{Indian Institute of Technology Hyderabad. Email:{\tt manaswi@cse.iith.ac.in}}
 \and
 Karthikeya Polisetty\footnote{Indian Institute of Technology Madras. Email: {\tt CS22B026@smail.iitm.ac.in}}
 \and
 Jayalal Sarma\footnote{Indian Institute of Technology Madras: Email:{\tt jayalal@cse.iitm.ac.in}}
 \and
 Nitin Saurabh\footnote{Indian Institute of Technology Hyderabad. Email:{\tt nitin@cse.iith.ac.in}}
}

\usepackage{amsmath,amsthm,xcolor,amssymb,enumitem}
\usepackage{thmtools}
\usepackage{thm-restate}
\usepackage{hyperref}
\hypersetup{colorlinks,linkcolor=blue,citecolor=blue}
\usepackage{cleveref}
\usepackage{fullpage,palatino,setspace}
\newtheorem{theorem}{Theorem}[section]
\newtheorem{definition}[theorem]{Definition}

\newtheorem{remark}[theorem]{Remark}
\newtheorem{proposition}[theorem]{Proposition}
\newtheorem{lemma}[theorem]{Lemma}
\newtheorem{claim}[theorem]{Claim}

\newtheorem{fact}[theorem]{Fact}
\newtheorem{conj}{Conjecture}

\newcommand{\zone}{\{0,1\}}
\newcommand{\zonep}{\{0,1,\ast\}}

\newcommand{\ket}[1]{|#1 \rangle}
\newcommand{\mc}[1]{\mathcal{#1}}
\newcommand{\abs}[1]{\left\lvert #1 \right\rvert}

\newcommand{\boolfn}[1]{\ensuremath {\zone}^{#1} \to \zone}
\newcommand{\s}{\mathsf{s}}
\newcommand{\bs}{\mathsf{bs}}
\newcommand{\fbs}{\mathsf{fbs}}

\newcommand{\cert}{\mathsf{C}}
\newcommand{\ucmin}{\mathsf{UC_{\min}}}
\newcommand{\uc}{\mathsf{UC}}
\newcommand{\adeg}{\widetilde{\deg}}
\newcommand{\rqc}{\mathsf{R}}
\newcommand{\qqc}{\mathsf{Q}}
\newcommand{\dqc}{\mathsf{D}}
\newcommand{\dqcs}[1]{\mathsf{D_{#1}}}
\newcommand{\cc}{\mathsf{cc}}
\newcommand{\mon}{\mathsf{mon}}

\newcommand{\fcs}{f_{\textup{CS}}}

\newcommand{\OR}{\mathsf{OR}}
\newcommand{\AND}{\mathsf{AND}}

\begin{document}

\maketitle

\begin{abstract}

For any Boolean function $f:\zone^n \to \zone$ with a complexity measure having value $k \ll n$, 
is it possible to restrict the function $f$ to $\Theta(k)$ variables while keeping the complexity preserved at $\Theta(k)$? 
Instantiation of this question for the measure of circuit complexity of the Boolean function was shown to be related to circuit lower bounds (Buresh-Oppenheim and Santhanam, 2006). Variants of the above question were also shown to have connections to the log-rank conjecture in communication complexity (Hrube\v{s}, 2024) and lower bounds in proof complexity (Razborov, 2016).
In the context of communication and query complexity, this question was recently studied by G{\"{o}}{\"{o}}s, Newman, Riazanov and Sokolov (2024).
They showed, among other results, that query complexity cannot be condensed losslessly. 

In this work, we show that there exists a Boolean function $f$ such that any restriction of $f$ to $O(\mathcal{M}(f))$ variables has $\mathcal{M}(\cdot)$-complexity at most $\tilde{O}(\mathcal{M}(f)^{2/3})$, where $\mathcal{M}$ is one of block sensitivity ($\bs$), fractional block sensitivity ($\fbs$), certificate complexity ($\cert$), deterministic query complexity ($\dqc$), zero-error randomized query complexity $(\rqc_0)$, and $\AND$ (and $\OR$)-decision tree query complexity. 
This improves upon the results of G{\"{o}}{\"{o}}s, Newman, Riazanov, and Sokolov (2024) for $\dqc$ and $\rqc_0$, and in particular answers their open question about the condensation of block sensitivity.

We complement the negative results on lossless condensation with positive results about lossy condensation. In particular, we show that for every Boolean function $f$ there exists a restriction of $f$ to $O(\mathcal{M}(f))$ variables such that its $\mathcal{M}(\cdot)$-complexity is at least $\Omega(\mathcal{M}(f)^{1/2})$, where $\mathcal{M} \in \{\bs,\fbs,\cert,\ucmin,\uc_1,\uc,\dqc,\adeg,\lambda\}$. In addition, we show lossy condensation for randomized and quantum query complexity with a slightly smaller exponent. 
\end{abstract}

\newpage

\section{Introduction}
\label{sec:intro}
The problem of proving circuit lower bounds is one of the most challenging problems in computational complexity theory. Buresh-Oppenheim and Santhanam \cite{BOS06} introduced the concept of \emph{hardness condensation} as a way to prove better circuit lower bounds. Their idea was to design a procedure that takes as input a hard to compute function $f$ and outputs another function $f'$, possibly on fewer input bits, such that the hardness of $f'$ with respect to its input length is greater than $f$. In particular, they showed hardness condensation for the complexity measure circuit-size. A decade later this strategy was successfully implemented by Razborov \cite{Razborov16} to establish strong lower bounds in proof complexity. Since then, the technique of hardness condensation has found many applications in proof complexity \cite{Razborov17,Razborov18,BN20,FPR22,BussThapen24,CD24,dRFJNP24}. 

Clearly, the paradigm of hardness condensation is quite general and one could ask the following question for any complexity measure $\mc{M}(\cdot)\colon$ 
\begin{quote}
    Given a function $f$ on $n$ variables with complexity $\mc{M}(f)$ 
    (a growing function of $n$), can we transform $f$ explicitly into another function $f'$ on a smaller number of variables $t$ such that $\mc{M}(f')=\Theta(\mc{M}(f))$ while also being maximal over all functions on $t$ variables?
\end{quote}

The aforementioned question has also been explored recently in the context of communication complexity and matrix ranks \cite{HHH22,Hrubes24a,GoosNR024}, and query complexity \cite{GoosNR024}. 
In the setting of communication complexity, we are given a $2^n \times 2^n$ Boolean matrix $M$ and the task is to solve the following two-player game: Alice is given a row $x\in [2^n]$ of $M$, Bob is given a column $y\in[2^n]$ of $M$, and their goal is to compute $M(x,y)$ while minimizing communication between them. 
Let $\cc(M)$ denote the deterministic communication complexity of this game. We say that $M$ \emph{condenses} to $k$ variables if $M$ contains a $2^k \times 2^k$ submatrix $N$ such that $\cc(N)=\Theta(k)$. Let us further say that the deterministic communication complexity \emph{condenses losslessly} if every $M$ condenses to $\Theta(\cc(M))$ variables. 
The hardness condensation question then asks whether $\cc(\cdot)$ condenses losslessly. Apart from being an interesting question about witnessing intrinsic hardness, it also has connections to the famous log-rank conjecture \cite{LS88}. 
In particular, the log-rank conjecture implies that there exists a universal constant $\alpha \geq 1$ such that every $M$ condenses to $\Omega({\cc(M)}^{1/\alpha})$ variables (see also \cite{Hrubes24a, GoosNR024}). In a recent work, Hrube\v{s} \cite{Hrubes24a} proved this implication unconditionally by showing that every $M$ condenses to $\Omega(\sqrt{\cc(M)})$ variables.   

In the context of query complexity, it was explored recently by G\"{o}\"{o}s, Newman, Riazanov and Sokolov \cite{GoosNR024}. They studied whether hardness condensation can be achieved using the simple and natural transformation - {\em restrictions} - that is, substituting variables by constants. Typically, to reduce the number of variables one does \emph{variable substitutions} and/or composes the variables of $f$ with \emph{small gadgets} (e.g., XOR, OR, Indexing functions).

Let $\mc{M}$ be any decision tree based complexity measure of Boolean functions (for example, sensitivity, block sensitivity, certificate complexity, query complexity, unambiguous certificate complexity, etc). 
We say that a Boolean function $f\colon\boolfn{n}$ \emph{condenses} to $k$ variables \emph{with respect to} the measure $\mc{M}$, if there exists a restriction $\rho\colon[n]\to\zonep$ with $|\rho^{-1}(\ast)|=k$, such that the restricted function $f|_\rho$ on $k$ variables has large $\mc{M}$, i.e., $\mc{M}(f|_\rho) = \Theta(k)$. We further say that $f$ \emph{condenses losslessly} with respect to $\mc{M}$ if it condenses to $\Theta(\mc{M}(f))$ variables. We also say that $\mc{M}$ condenses losslessly if every function $f$ condenses losslessly w.r.t.~$\mc{M}$.

For example, consider the measure of query complexity of Boolean functions, which is studied by \cite{GoosNR024}. For $f\colon\boolfn{n}$, the query complexity $\dqc(f)$ of $f$ is the minimum number of queries made by a deterministic decision tree algorithm computing $f$ in the worst-case over all inputs $x \in \{0,1\}^n$. 
Following the above definition (also from \cite{GoosNR024}), we say that $f$ condenses to $k$ variables with respect to the measure of deterministic query complexity if there exists a partial assignment $\rho\colon [n]\to\zonep$ with $|\rho^{-1}(\ast)|=k$, such that the restricted function $f|_\rho$ on $k$ variables has large query complexity, i.e., $\dqc(f|_\rho) = \Theta(k)$. The hardness condensation question, w.r.t.~$\dqc(\cdot)$, then asks if $\dqc(\cdot)$ condenses losslessly.

The authors of \cite{GoosNR024} answered the aforementioned question in the negative. 
That is, they showed a function $f$ such that for every restriction $\rho$ that fixes all but at most $O(\dqc(f))$ variables, the query complexity $\dqc(f|_\rho)$ of the restricted function $f|_\rho$ is $\tilde{O}(\dqc(f)^{3/4})$ (where $\tilde{O}$ hides polylog factors in $\dqc(f)$). 
In other words, deterministic (and even randomized) query complexity cannot be condensed losslessly in general. The authors wondered if there is a function that witnesses greater loss, i.e., for which the hardness is even less condensible.  
Observe that the best possible bound cannot be better than $O(\dqc(f)^{1/3})$.
Indeed, it follows from the fact that every $f$ condenses to $\deg(f)$ variables\footnote{Consider a monomial $\mathfrak{m}$ of maximum degree $\deg(f)$ in the unique real polynomial representing $f$. Let $\rho$ be a restriction that fixes all variables except the ones in $\mathfrak{m}$. Then  $\dqc(f|_\rho) = \deg(f|_\rho)=\deg(f)$.} and $\deg(f) \geq \dqc(f)^{1/3}$ \cite{Midri04}.
In \cite{GoosNR024}, the authors left open, among other things, the problem of closing this gap and verifying lossless condensation for other decision tree based measures like, block sensitivity, certificate complexity, unambiguous certificate complexity, etc. 

In this article, we delve deeper into the question of hardness condensation for $\dqc(\cdot)$ and initiate the study of hardness condensation for almost all other decision tree based measures. In particular, we close the gap for $\dqc$ from both sides, showing that it is not possible to condense even with an exponent of $2/3+\varepsilon$ (improving from $3/4$) and it is always possible to condense with an exponent of $1/2$ (improving from $1/3$). We show similar results for many other well known complexity measures like block sensitivity, certificate complexity, fractional block sensitivity, $\AND$-decision tree complexity, etc.~(Theorems~\ref{thm:bs-cert-condensation-intro}, \ref{thm:query-condensation-intro}, \ref{thm:or-condensation-intro} and \ref{thm:exponent-lb}).

Unfortunately, except for approximate degree and spectral sensitivity, there are gaps in terms of what can and cannot be achieved for these complexity measures in terms of hardness condensation (see Table~\ref{table:complexity-measures}). These gaps leave many intriguing questions open; we hope that future research will be able to resolve these questions.

Before going into the details of our results and proof techniques we present some possible applications of hardness condensation by restrictions. 

\paragraph*{Applications of Hardness Condensation by Restrictions:}
We now list down a few applications of the hardness condensation by restrictions.\\[-2mm]

\noindent{\bf Log-Rank Conjecture:} We begin with an application to the log-rank conjecture highlighted by \cite{hart2025condensing}. 
Let the $\AND$-decision tree query complexity $\dqcs{\wedge}(f)$ of a Boolean function $f$ be the minimum number of $\AND$ queries to the input made by a deterministic query algorithm computing $f$ in the worst case (see \cref{defi:adt}).  

For any $f\colon\boolfn{n}$, let $\mon(f)$ denote the number of monomials with non-zero coefficients in the unique multilinear polynomial representing $f$. Further, define the communication problem $f_{\wedge}\colon\zone^n\times\zone^n \to\zone$ as follows $f_{\wedge}(x,y)= f(x_1\wedge y_1,\ldots ,x_n\wedge y_n)$. It then follows that the rank of the communication matrix of $f_\wedge$ equals $\mon(f)$ \cite{BdW01}.  Moreover, it is known that \[\cc(f_\wedge) \leq 2\dqcs{\wedge}(f) = O(\log^{5}\mon(f)\cdot\log n)=O(\log^5\mathsf{rank}(f_\wedge)\cdot \log n),\] where the second inequality was shown by \cite{KLMY21}. That is, the log-rank conjecture holds for $f_\wedge$ as long as $\log\log\mon(f)=\Omega(\log\log n)$. In other words, the log-rank conjecture remains unsolved for $f_\wedge$ when $f$ has a very sparse representation as polynomial. To resolve this, \cite{hart2025condensing} observed that a hardness condensation of the following form suffices. 
\begin{conj}[\cite{hart2025condensing}]
    There exist universal constants $\alpha, \beta >0$ such that for every  $f\colon\boolfn{n}$ there exists a restriction $\rho$ with $|\rho^{-1}(\ast)|=2^{O(\log^\alpha \mon(f))}$ and $\dqcs{\wedge}(f|_\rho) = \Omega(\dqcs{\wedge}(f)^{\beta})$.  
\end{conj}

\noindent{\bf Separations between complexity measures:} We now illustrate 
an application of hardness condensation to witness improved separations between complexity measures. We illustrate the argument with a simple example to highlight the applicability of this approach. 

Consider the modified Rubinstein function $f$ given in \cref{defi:modified-rubinstein}. It can be easily shown that $\dqc(f)=k^2$ and $\bs(f)=k^{1.5}$ (see also \cref{lem:bounds-mod-rub}). Thus we have an example where $\dqc(f) = \Omega(\bs(f)^{4/3})$ and $\bs(f)=\Omega(\s(f)^{3/2})$. Using hardness condensation, we now amplify the gap between $\dqc$ and $\bs$. We know from \cref{thm:bs-cert-condensation-intro} that $f$ witnesses the incondensability of $\bs$ with exponent $2/3$. That is, for every restriction $\rho\colon[n]\to\zonep$ with $|\rho^{-1}(\ast)|=O(\bs(f))$ we have $\bs(f|_\rho)=O(k)=O(\bs(f)^{2/3})$. Now consider a restriction $\gamma$ with $|\gamma^{-1}(\ast)|=\Theta(\bs(f))$ such that $\dqc(f|_{\gamma})=|\gamma^{-1}(\ast)|=\Theta(\bs(f))$. Such a restriction is easily seen to exist since $\dqc(f)$ is full to begin with. We then have by \cref{thm:bs-cert-condensation-intro} that $\bs(f|_\gamma)=O(k)=O(\bs(f)^{2/3})$. We therefore have obtained a function $h:=f|_\gamma$ such that $\dqc(h) = \Omega(\bs(h)^{3/2})$ improving over $4/3$ separation witnessed by $f$. So in general such an approach can be used to find better separations. 

A candidate for a fruitful application of the approach outlined above could be the gap between $\dqc$ and $\rqc$. The best known separation between $\dqc$ and $\rqc$ is quadratic~\cite{ABBL+17}. Note that we know of a function that witness more than a quadratic gap between $\rqc$ and $\s$ \cite{DHT17}. The approach outlined above exploits this gap with respect to sensitivity to obtain improved separations between complexity measures.  

We further note that the positive condensation result in \cref{thm:exponent-lb} also shows a limit to amplification using the approach outlined above. For example, we have shown in \cref{thm:exponent-lb} that for every $f$ there exists a restriction $\rho\colon[n]\to\zonep$ with $|\rho^{-1}(\ast)|=O(\bs(f))$ such that $\bs(f|_\rho)=\Omega(\sqrt{\bs(f)})$. Therefore, our approach cannot show a gap of more than $1/2$ between $\dqc$ and $\bs$.

\subsection{Our Results}

In this work, we further explore the phenomenon of hardness condensation in the setting of query complexity and other Boolean function parameters using restrictions. 

We begin by recalling 
that sensitivity and degree are two such measures that condense losslessly. As mentioned above, \cite{GoosNR024} shows that deterministic query complexity cannot be condensed losslessly. 
Our first result (proved in Section~\ref{sec:negative-I}) shows that block sensitivity, fractional block sensitivity and certificate complexity cannot be condensed losslessly.
\begin{restatable}{theorem}{bsCondensation}
    \label{thm:bs-cert-condensation-intro}
    There exists a Boolean function $f\colon\boolfn{n}$ with 
    $\bs(f)=\fbs(f)=\cert(f)=n^{3/4}$, such that for every restriction $\rho\colon[n]\to\zonep$ with $|\rho^{-1}(\ast)|= O(n^{3/4})$ we have $\bs(f|_{\rho}) \leq \fbs(f|_{\rho}) \leq \cert(f|_{\rho}) = O(\sqrt{n})= O({\bs(f)}^{2/3})$. 
\end{restatable}

Our second result (proved in Section~\ref{sec:negative-II}) shows that deterministic query complexity is even less condensible than shown in \cite{GoosNR024}. We also show that \textsf{AND}-decision tree complexity  and $0$-decision tree complexity cannot be condensed losslessly.

\begin{restatable}{theorem}{DCondensation}
    \label{thm:query-condensation-intro}
    There exists a Boolean function $f\colon\boolfn{n}$ with $\Tilde{\Omega}(n^{3/5}) =\dqcs{0}(f)\leq\dqcs{\wedge}(f)\leq\dqc(f)=O(n^{3/5})$, such that for every restriction $\rho\colon[n]\to\zonep$ with $|\rho^{-1}(\ast)|= O(n^{3/5})$ we have $\dqcs{0}(f|_\rho)\leq \dqcs{\wedge}(f|_\rho) \leq \dqc(f|_{\rho})=\Tilde{O}(n^{2/5})=\Tilde{O}({\dqcs{0}(f)}^{2/3})$. 
\end{restatable}

Our proof technique for \cref{thm:query-condensation-intro} works even in the dual setting
and therefore we also obtain the following incondensability of $\OR$-decision tree complexity. 
\begin{restatable}{theorem}{DCondensationOR}
    \label{thm:or-condensation-intro}
    There exists a Boolean function $f\colon\boolfn{n}$ with $\Tilde{\Omega}(n^{3/5}) =\dqcs{1}(f)\leq\dqcs{\vee}(f)=O(n^{3/5})$, such that for every restriction $\rho\colon[n]\to\zonep$ with $|\rho^{-1}(\ast)|= O(n^{3/5})$ we have $\dqcs{1}(f|_\rho)\leq \dqcs{\vee}(f|_\rho)=\Tilde{O}(n^{2/5})=\Tilde{O}({\dqcs{1}(f)}^{2/3})$. 
\end{restatable}

We further observe that the zero-error randomized query complexity of the example in \cref{thm:query-condensation-intro} is also high $\Tilde{\Omega}(n^{3/5})$, thereby obtaining a similar incondensability result for $\rqc_0$. 
\begin{restatable}{theorem}{R0Condensation}
    \label{thm:R0-condensation-intro}
    There exists a Boolean function $f\colon\boolfn{n}$ with $\rqc_0(f)=\Tilde\Theta(n^{3/5})$, such that for every restriction $\rho\colon[n]\to\zonep$ with $|\rho^{-1}(\ast)|= O(\rqc_0(f))$ we have $\rqc_0(f)=\Tilde{O}({\rqc_0(f)}^{2/3})$.
\end{restatable}

We note that Theorems~\ref{thm:bs-cert-condensation-intro} and \ref{thm:query-condensation-intro} answer (or, partially answer)  Problems~$1$ and $2$ in \cite{GoosNR024}. \cref{thm:query-condensation-intro} also improves upon (one of) the results in \cite{GoosNR024}. 
It is natural to wonder if the gaps in Theorems~\ref{thm:bs-cert-condensation-intro} and \ref{thm:query-condensation-intro} are quantitatively optimal.  That is, does there exist a function for which a measure is even less condensible than established here?

In particular, for deterministic query complexity, it asks if there exists a function $f$ such that for all restrictions $\rho$ with $|\rho^{-1}(\ast)|=O(\dqc(f))$, we have $\dqc(f|_\rho) = O(\dqc(f)^\alpha)$ where $\alpha < 2/3$? 
Observe that $\alpha \geq 1/3$, since $\dqc(f)\leq\deg(f)^3$ \cite{Midri04} and $\deg(f)$ condenses losslessly.  
Similarly, for block sensitivity, the exponent $\alpha \geq 1/4$, since $\bs(f)\leq \s(f)^4$ \cite{Huang} and $\s(f)$ condenses losslessly. 

Observe that the lower bound on the exponent $\alpha$ can be improved if conditionally based on other relationships between the parameters. 

For example, if $\dqc(f)\leq\deg(f)^2$, or $\bs(f)\leq \s(f)^2$, then such an improvement is possible. 
Our third result (proved in Section~\ref{sec:positive-results}) establishes the same lower bound on the exponent $\alpha$ unconditionally. 
\begin{restatable}{theorem}{PositiveResults}
     \label{thm:exponent-lb}
    Let $\mc{M} \in \{ \bs,\fbs,\cert,\ucmin,\uc_1,$ $\uc,\dqc,\adeg,\lambda,\rqc,\rqc_0,\qqc,\qqc_E\}$ be a complexity measure and $f\colon\boolfn{n}$ be a Boolean function. Then, the following holds:
    \begin{enumerate}[label=\alph*)]
        \item \label{thm4:a} For $\mc{M}\notin\{\rqc,\rqc_0,\qqc,\qqc_E\}$, there exists a restriction $\rho\in\zonep^{n}$ with $|\rho^{-1}(\ast)| = \Theta(\mc{M}(f))$ such that $\mc{M}(f|_\rho) = \Omega(\sqrt{\mc{M}(f)})$.  
        \item \label{thm4:b} For $\mc{M}\in\{\rqc,\rqc_0,\qqc_E\}$, there exists a restriction $\rho\in\zonep^{n}$ with $|\rho^{-1}(\ast)| = \Theta(\mc{M}(f))$ such that $\mc{M}(f|_\rho) = \Omega({\mc{M}(f)}^{1/3})$.
        \item \label{thm4:c} For $\mc{M}=\qqc$, there exists a restriction $\rho\in\zonep^{n}$ with $|\rho^{-1}(\ast)| = \Theta(\qqc(f))$ such that $\qqc(f|_\rho) = \Omega({\qqc(f)}^{1/4})$.
    \end{enumerate}
\end{restatable}

Table~\ref{table:complexity-measures} lists\footnote{A `$*$' entry represents that we do not know whether the measure labeling the row condenses losslessly or not. An entry $\gamma$ in the Negative Result column represents that there exists a function $f$ such that for every restriction $\rho\in\zonep^{n}$ with $|\rho^{-1}(*)|= O(\mathcal{M}(f))$ we have $\mathcal{M}(f|_{\rho}) = O(\mathcal{M}(f)^{\gamma})$, where $\mathcal{M}$ is the complexity measure labeling the row. An entry $\xi$ in the Positive Result column represents that for every function $f$ there exists a restriction $\rho\in\zonep^{n}$ with $|\rho^{-1}(\ast)| = \Theta(\mathcal{M}(f))$ such that $\mathcal{M}(f|_\rho) = \Omega({\mathcal{M}(f)}^{\xi})$, where $\mathcal{M}$ is the complexity measure labeling the row.}
the known bounds on condensation for many decision tree measures. Note that for approximate degree ($\adeg$) and spectral sensitivity ($\lambda$), we have tight bounds on lossy condensation that can be achieved. For the other measures, there is a gap in the negative result and the positive result. Thus, the understanding for many of them is not complete, leaving many interesting open questions. 
\begin{table}[!h]
    \centering
    \begin{tabular}{|c||c|c|}
        \hline
$\textbf{Complexity Measure}$ & $\textbf{Negative Result}$ & $\textbf{Positive Result}$ \\
\hline
{$\dqc$} & $2/3$ (Theorem~\ref{thm:query-condensation-intro}) &  $1/2$ (Theorem~\ref{thm:exponent-lb}) \\
\hline
{$\dqcs{0}$, $\dqcs{1}$, $\dqcs{\land}$, $\dqcs{\lor}$} & $2/3$ (Theorems~\ref{thm:query-condensation-intro} \& \ref{thm:or-condensation-intro}) & $\ast$ \\
\hline
{$\rqc_0$ } & $2/3$ (\cref{thm:Rquery-condensation}) & $1/3$ (Theorem~\ref{thm:exponent-lb}) \\
\hline
{$\rqc$} & $3/4$ (\cite{GoosNR024}) &  $1/3$ (Theorem~\ref{thm:exponent-lb}) \\
\hline
$\cert, \fbs, \bs$ & $2/3$ (Theorem~\ref{thm:bs-cert-condensation-intro}) & $1/2$ (Theorem~\ref{thm:exponent-lb}) \\ 
\hline
{$\uc, \ucmin$} & * &  $1/2$ (Theorem~\ref{thm:exponent-lb}) \\
\hline
{$\qqc_{E}$} & * & $1/3$ (Theorem~\ref{thm:exponent-lb}) \\
\hline
{$\qqc$} & $1/2$ [$\OR_n$] & $1/4$ (Theorem~\ref{thm:exponent-lb})\\
\hline
$\adeg$ & $1/2$ [$\OR_n$] & $1/2$ (Theorem~\ref{thm:exponent-lb}) \\
\hline
{$\lambda$} & $1/2$ [$\OR_n$] & $1/2$ (Theorem~\ref{thm:exponent-lb})  \\
\hline
$\s$ & -- & condenses losslessly (see also Lemma~\ref{lem:s-deg-stronger}) \\
\hline
{$\deg$} & -- &  condenses losslessly (see also Lemma~\ref{lem:s-deg-stronger})  \\
\hline
    \end{tabular}
    \caption{Hardness condensation by variable restrictions}
    \label{table:complexity-measures}
\end{table}

\subsection{Proof Ideas} 
In this subsection, we explain the ideas and techniques used to obtain our results. \\[-2mm]

\noindent{\bf Incondensability of $\bs$ and $\cert$: } Observe that since $\s \leq \bs\leq \cert$ and sensitivity condenses losselessly, a counter-example $f\colon\boolfn{n}$ to lossless condensation of block sensitivity must witness a gap between $\bs$ and $\s$, i.e., $\bs(f)=\omega(\s(f))$. Simultaneously, there should also be a gap between the number of variables $n$ and $\bs(f)$, so that every restriction $\rho$ with $|\rho^{-1}(\ast)|=O(\bs(f))$ must set some variables to constants. Thus the Rubinstein function (and rather a modification of it to introduce a gap between $n$ and $\bs$) becomes a natural choice. 

We then observe that the Rubinstein function also has a nice property that only a particular kind of inputs (all $0$'s) have high block sensitivity and certificate complexity. We use this structure to argue that the certificate complexity, and in turn the block sensitivity, must reduce after restriction. \\[-2mm] 

\noindent{\bf Incondensability of $\dqc$ and its variants:} 
Incondensability of $\dqc$ was shown by \cite{GoosNR024} using a cheat-sheet version of Tribes ($\AND$ of $\OR$s). Our improved example is also a cheat-sheet version $\fcs$ (\cref{eq:cheat-sheet-tribes}) of Tribes $f$ (\cref{defi:tribes}). However we choose different arities for $\AND$ and $\OR$s. In particular, we choose $f=\AND_k \circ \OR_{\sqrt{k}}$. This turns out to be a crucial insight that helps design an improved query algorithm for $\fcs|_\rho$; thereby obtaining an improvement. 

Our query algorithm for $\fcs|_\rho$, like in \cite{GoosNR024}, has two parts. In the first part our algorithm finds a string $z\in\zone^{\log t}$, where $t$ is the number of cheat-sheet cells, that the address part $f(x_1|_\rho)f(x_2|_\rho)\cdots f(x_{\log t}|_\rho)$ can take if the restricted function $\fcs|_\rho$ were to evaluate to $1$. 
In the second part, the algorithm simply reads the cell pointed to by $z$ and verifies it. 

In the first part, the algorithm makes queries with an aim to find each bit of $z=z_1z_2\cdots z_{\log t} \in \zone^{\log t}$, while satisfying the following property: either we know the value of $z_j=f(x_j|_\rho)$ correctly or there are no valid cheat-sheet cells with $z_j=1$. For each $j\in[\log t]$, the query algorithm maintains the following two sets: $(a)$ the set $\mc{O}_0$ of $\OR$s in the $j$-th copy of $f$ that can possibly evaluate to $0$, and $(b)$ the set $\mc{V}_1$ of cheat-sheet cells with $z_j=1$ and not found to be inconsistent with the queries made so far. Note that initially $|\mc{O}_0|=k$ and $|\mc{V}_1|=t/2$. Now the algorithm makes queries in such a way that each query reduces either $|\mc{O}_0|$ by $1$ or $|\mc{V}_1|$ by a factor of $(1-\frac{1}{2\sqrt{k}})$. Finally when the algorithm stops, after making at most $O(k+\sqrt{k}\log t)$ queries, we either know the value of $f(x_j|_\rho)$ or more crucially are guaranteed to be in one of the following two cases: $(i)$ $|\mc{O}_0| = O(\sqrt{k})$ or $(ii)$ $|\mc{V}_1|=O(k)$. In case $(i)$ the algorithm can query all remaining $\OR$s in $\mc{O}_0$ to find out the value of $f(x_j|_\rho)$, while in case $(ii)$ it uses $O(\log k)$ queries to find inconsistency to remove a cell from $\mc{V}_1$ (\cref{claim:discarding-cheatsheets}), and finally the last remaining cheat-sheet cell in $\mc{V}_1$ is verified for $1$-certificate that it claims to contain for $f(x_j|_\rho)$.  

To find a query that allows us to reduce either $|\mc{O}_0|$ by $1$ or $|\mc{V}_1|$ by a factor of $(1-\frac{1}{2\sqrt{k}})$, we use the fact that $|\rho^{-1}(\ast)|=O(\dqc(\fcs)=\Tilde{O}(k^{1.5})$ and thus there must be a claimed $1$-certificate that is completely revealed in at least half the fraction of $\mc{V}_1$. \\[-2mm]

\noindent{\bf Lossy condensation:} We complement our negative results on lossless condensation with positive results about lossy condensation. Our proofs are analysis of different cases based on whether $\s(f)$ is large or $\deg(f)$ is large or neither. In each case we find a restriction that gives the desired lossy condensation. In the first two cases, since sensitivity and degree condenses losslessly (\cref{lem:s-deg-stronger}), we easily find a restriction witnessing lossy condensation. The last and the interesting case is when neither of them is large; we then argue that it must be the case that block sensitivity is large,  which then helps us find a restriction that witnesses lossy condensation. For example, consider the case of the deterministic query complexity. As alluded before, the interesting case is when neither sensitivity nor degree is large, i.e., at least $\sqrt{\dqc(f)}$. Then, from $\dqc(f) \leq \deg(f)\bs(f)$ \cite{Midri04}, it follows that $\bs(f) \geq \sqrt{\dqc(f)}$. Thus we can consider the following restriction $\rho$ based on an input $z$ with the maximum block sensitivity: leave $\sqrt{\dqc(f)}$ disjoint (minimal) sensitive blocks unset while setting the rest of the variables according to $z$. Since $\s(f) < \sqrt{\dqc(f)}$, $|\rho^{-1}(\ast)| = O(\dqc(f))$. At the same time, by construction, we also have $\dqc(f|_\rho) \geq \bs(f|_\rho) \geq \sqrt{\dqc(f)}$.\\[-2mm]

\noindent{\bf Organization of the paper:} The necessary preliminaries are presented in \cref{sec:prelims}. We present the incondensability of block sensitivity and certificate complexity (\cref{thm:bs-cert-condensation-intro}) in \cref{sec:negative-I}. The incondensability of $\AND$(and $\OR)$-decision tree query complexity and improved examples for the incondensability of the deterministic query complexity and its variants (Theorems~\ref{thm:query-condensation-intro}, \ref{thm:or-condensation-intro} and \ref{thm:R0-condensation-intro}) are presented in \cref{sec:negative-II}. We present the positive results about lossy condensation in \cref{sec:positive-results}. We also show that \cref{thm:bs-cert-condensation-intro} is optimal for modified Rubinstein functions in Appendix~\ref{sec:apx-rub}.

\section{Preliminaries}
\label{sec:prelims}

We use $[n]$ to denote the set $\{1,2,\ldots ,n\}$. We recall definitions of complexity measures that we will be using throughout. We refer the reader to the survey~\cite{BW} for an introduction to the complexity of Boolean functions and complexity measures.
Several additional complexity measures and their relations among each other can also be found in~\cite{DHT17} and~\cite{ABK+}. 

In the following $f\colon\boolfn{n}$ is a Boolean function on $n$ variables. For an input $x\in\zone^n$ and $i\in[n]$, let $x^i$ represent the string in $\zone^n$ that is obtained from $x$ by flipping (or, negating) the $i$-th bit. We say that the $i$-th bit is \emph{sensitive} at $x$ if $f(x)\neq f(x^i)$. 
\begin{definition}[Sensitivity]
\label{defi:sens}
The \emph{sensitivity} of $f$  on an input $x$, $\s(f,x)$, is
defined as the number of sensitive bits at $x$ with respect to $f$. That is, $\s(f,x) = |\{i\in[n] \mid f(x) \neq f(x^i)\}|$. 
Then, the \emph{sensitivity} of $f$, denoted $\s(f)$, is defined as $\s(f)=\max\{\s(f,x)\mid x\in \zone^n\}$. 

We also define $0$-\emph{sensitivity} of $f$ as
$\s_0(f) = \max\{\s(f,x)\mid x\in \zone^n, f(x)=0\}$ and $1$-\emph{sensitivity} of $f$ as $\s_1(f) = \max\{\s(f,x)\mid x\in \zone^n, f(x)=1\}$.
\end{definition}
Similarly for $B\subseteq [n]$, let $x^{B}$ denote the string obtained from $x$ by flipping the bits indexed by the block $B$. We say that the block $B$ is \emph{sensitive} at $x$ if $f(x)\neq f(x^{B})$. 
\begin{definition}[Block Sensitivity]
\label{defi:b-sens}
The \emph{block sensitivity} of a function $f$ on an input $x$, $\bs(f,x)$, is the maximum number of disjoint subsets $B_1, B_2, \dots , B_r$ of
$[n]$ such that for all $j\in[r]$, $f(x) \neq f(x^{B_j})$. Then, the \emph{block sensitivity} of $f$, denoted $\bs(f)$, is $ \max \{ \bs(f,x)\mid x \in \zone^n\}$.
Similarly, define $\bs_0(f) = \max\{\bs(f,x)\mid x\in \zone^n, f(x)=0\}$ and $\bs_1(f) = \max\{\bs(f,x)\mid x\in \zone^n, f(x)=1\}$. 
\end{definition}
We say that a sensitive block at $x$ is \emph{minimal} if no proper subset of it is a sensitive block at $x$. We then have the following bound on the size of a minimal sensitive block.
\begin{lemma}[\cite{nisan1989crew}]
    \label{lem:minimal-sens-block}
    Let $B$ be a minimal sensitive block at an input $x$. Then $|B|\leq \s(f)$.
\end{lemma}
We also need the following linear programming based generalization of block sensitivity. 
\begin{definition}[Fractional Block Sensitivity]
\label{defi:fbs}
Let $\mathcal{W}(f,x):=\{B\subseteq [n]\mid f(x^{B})\neq f(x)\}$ denote the set of all sensitive blocks for an input $x\in\zone^{n}$. The \emph{fractional block sensitivity} of $f$ at $x$, denoted $\fbs(f,x)$, is the value of the linear program, 
\begin{align*}
    \text{maximize } & \sum_{B\in\mathcal{W}(f,x)}y_B,   \\
    \text{subject to } & \sum_{B\colon i\in B}y_B \leq 1,\quad  \forall\, i \in [n], \\
    &\quad  0 \leq y_B \leq 1,\quad  \forall\, B \in \mathcal{W}(f,x). 
\end{align*}
The \emph{fractional block sensitivity} $\fbs(f)$ of $f$ is then defined as \(\fbs(f):=\underset{x\in\zone^{n}}{\max}\fbs(f,x)\).
\end{definition}
A \emph{partial assignment} is a function $\rho\colon \zone^n \to \{0,1,*\}$ where the size of $\rho$ is $n-|\rho^{-1}(*)|$. Let $S \subseteq [n]$ denote the set of indices where $\rho$ is not $*$. We say that an input $x\in\zone^n$ is \emph{consistent} with $\rho$ iff $x_i=\rho(i)$ for all $i\in S$. For $b\in\zone$, a $b$-\emph{certificate} for the Boolean function $f$ is a partial assignment $\rho$ such that for all inputs $x$ and $y$ consistent with $\rho$ we have $f(x)=f(y)=b$.
\begin{definition}[Certificate Complexity]
\label{defi:cert}
The \emph{certificate complexity} of a function $f$ on an input $x$, denoted $\cert(x, f)$, is the size of the smallest $f(x)$-certificate that is consistent with $x$. The \emph{certificate complexity} of $f$, denoted $\cert(f)$, is $\max_x\{\cert (f,x)\}$. Further, define $\cert_{0}(f) = \max_{x\colon f(x)=0}\{\cert(f,x)\}$ 
and $\cert_{1}(f) = \max_{x\colon f(x)=1}\{\cert(f,x)\}$. 
\end{definition}
We now note the following relation between these measures.
\begin{lemma}[\cite{nisan1989crew,Tal13}]
    \label{lem:s-bs-c-relation}
    \(s(f,x)\leq\bs(f,x)\leq\fbs(f,x)\leq\cert(f,x)\leq\bs(f,x)\s(f).\)
\end{lemma}
We also require the following generalization of certificate complexity.
\begin{definition}[Unambiguous Certificate Complexity]
\label{defi:uc}
Fix $b\in\zone$. A set $U$ of partial assignments is said to form an \emph{unambiguous} collection of $b$-certificates for $f$ if
\begin{itemize}
    \item each partial assignment in $U$ is a $b$-certificate for $f$,
    \item for each $x\in f^{-1}(b)$ there is some $p\in U$ that is consistent with $x$, and 
    \item no two partial assignments in $U$ are consistent.
\end{itemize}
We then define $\uc_{b}(f)$ to be the minimum value
of $\max_{p\in U} |p|$ over all choices of such collections $U$. 
Further, define $\uc(f):= \max\{\uc_{0}(f),\uc_{1}(f)\}$ and $\ucmin(f):= \min\{\uc_{0}(f),\uc_{1}(f)\}$. 
\end{definition}

We now define the query model of computation. In this model, the input bits can be accessed by queries to an input oracle and the complexity of computing a Boolean function $f$ is the number of queries made to the oracle.
\begin{definition}[Deterministic Decision Trees]
    \label{defi:ddt}
    Let $f\colon\boolfn{n}$ be a Boolean function. Let $A$ be a deterministic algorithm that computes $f$ by making queries to the bits of an unknown input $x$. Further, let $\text{cost}(A,x)$ be the number of queries made by $A$ on the input $x$, and $\text{cost}(A)=\max_x\text{cost}(A,x)$. Then, the deterministic decision tree complexity (also known as query complexity) of $f$, denoted $\dqc(f)$, is $\min_A\{\text{cost}(A)\}$, where the minimum is over all deterministic algorithm computing $f$.   
\end{definition}
We note the following easy to observe relations.
\begin{fact}
\label{fact:uc-d-relation}
\(\deg(f)\leq \ucmin(f)\leq \uc_1(f) \leq \uc(f) \leq \dqc(f)\), and\, \(\bs(f)\leq\cert(f)\leq \uc(f)\).
\end{fact}
We also need the following variants of deterministic decision tree. 
\begin{definition}[\textsf{AND}-Decision Trees]
    \label{defi:adt}
    Let $f\colon\boolfn{n}$ be a Boolean function. Let $A$ be a deterministic algorithm that computes $f$ by querying $\AND$ of a subset $S\subseteq [n]$ of the bits of an unknown input $x$. Further, let $\text{cost}(A,x)$ be the number of queries made by $A$ on the input $x$, and $\text{cost}(A)=\max_x\text{cost}(A,x)$. Then, the $\AND$-decision tree complexity of $f$, denoted $\dqcs{\wedge}(f)$, is $\min_A\{\text{cost}(A)\}$, where the minimum is over all deterministic $\AND$-decision tree algorithm computing $f$.   
\end{definition}
We can similarly define $\OR$-decision trees which is allowed to query $\OR$ of a subset of input variables. 
\begin{definition}[\textsf{OR}-Decision Trees]
    \label{defi:ordt}
    The $\OR$-decision tree complexity of $f$, denoted $\dqcs{\vee}(f)$, is the minimum number of $\OR$ queries made by a deterministic decision tree algorithm computing $f$ in the worst-case.   
\end{definition}
We also need the notion of decision tree algorithms that minimize the the number of $0$-queries (or, $1$-queries) while computing $f$.  
\begin{definition}[$0$-Decision Trees]
    \label{defi:zerodt}
    Let $f\colon\boolfn{n}$ be a Boolean function. Let $A$ be a deterministic algorithm that computes $f$ by making queries to the bits of an unknown input $x$. Further, let $\text{cost}(A,x)$ be the number of queries made by $A$ on the input $x$ that are answered as $0$, and $\text{cost}(A)=\max_x\text{cost}(A,x)$. Then, the deterministic $0$-decision tree complexity of $f$, denoted $\dqcs{0}(f)$, is $\min_A\{\text{cost}(A)\}$, where the minimum is over all deterministic $0$-decision tree algorithm computing $f$.   
\end{definition}
\begin{definition}[$1$-Decision Trees]
    \label{defi:onedt}
    The $1$-decision tree complexity of $f$, denoted $\dqcs{1}(f)$, is the minimum number of $1$-queries made by a deterministic decision tree algorithm computing $f$ in the worst-case.   
\end{definition}
We note the following easy to observe relations between these variants (see also \cite{LM19}). 
\begin{fact}
    \label{fact:zerodt-anddt-dt}
    \(\dqcs{0}(f) \leq \dqcs{\wedge}(f) \leq \dqc(f)\), and\, \(\dqcs{1}(f)\leq \dqcs{\vee}(f)\leq \dqc(f)\).
\end{fact}

We now define randomized decision trees which, in short, is just a probability distribution over deterministic decision trees. 
\begin{definition}[Randomized Decision Trees]
    \label{defi:rdt}
    A randomized query algorithm $A$ is said to compute $f$ with error at most $1/3$ (or, bounded-error), if $\Pr[A(x)=f(x)] \geq 2/3$ for all $x\in\zone^n$, where the probability is taken over the (internal) randomness of $A$. We then define $\rqc(f)$ to be the minimum number of queries required by any randomized query algorithm to compute $f$ with error at most $1/3$. 

    We will also consider a \emph{zero-error} model in the randomized case. We say that a randomized decision tree computes $f$ with zero-error if it never gives an \emph{incorrect} answer, but it may output ``don't know'', with probability at most $1/2$ for every input. Let $\rqc_0(f)$ denote the minimum number of queries required by any zero-error randomized algorithm computing $f$.  
\end{definition}
We also define the model of quantum query algorithms. 
\begin{definition}
    \label{defi:qdt}
    A $T$-query quantum algorithm has the form $A=U_TO_xU_{T-1}\cdots O_xU_1O_xU_0$, where the $U_k$'s are fixed unitary transformations independent of the input $x$. A query is the application of the unitary transformation $O_x$ that maps 
    \[O_x : \ket{i,b,z} \mapsto \ket{i,b\oplus x_i,z},\] 
    where $i\in [n]$, $b\in\zone$, and $x\in\zone^n$ is an unknown input. The $z$-part corresponds to the workspace, which is not affected by the query.    
    The final state $A\ket{0}$ depends on $x$ via the $T$ applications of $O_x$. The output of the algorithm is determined by measuring the rightmost qubit of the final state. 

    We say that a quantum query algorithm computes $f$ \emph{exactly} (resp., \emph{with bounded-error}) if its output equals $f(x)$ with probability $1$ (resp., at least $2/3$), for all $x\in\zone^n$. Then, define $\qqc_{E}(f)$ (resp., $\qqc(f)$) to be the minimum number of queries made by a quantum query algorithm computing $f$ exactly (resp., with bounded-error).   
\end{definition}
We note the following easy to observe relations.
\begin{fact}
    \label{fact:q-r-d-relations}
    \(\qqc(f)\leq \rqc(f) \leq \rqc_0(f) \leq \dqc(f)\), and\, \(\qqc(f)\leq\qqc_E(f)\leq \dqc(f)\). 
\end{fact}

A real multilinear polynomial $p(x_1,\ldots ,x_n)$ is said to \emph{represent} a Boolean function $f\colon\boolfn{n}$, if $p(x)=f(x)$ for all $x\in\zone^n$. 
It is easily seen that a Boolean function $f$ has a unique multilinear polynomial $p$ representing it. Let $\deg(f)$ denote the degree of this unique multilinear polynomial representing $f$. We also need the following notion of approximating polynomials. 
\begin{definition}
    \label{defi:adeg}
    We say that a polynomial $p$ \emph{approximates} a Boolean function $f$, if $|p(x)-f(x)|\leq 1/3$ for all $x\in\zone^n$. Let $\adeg(f)$ denote the minimal possible degree of a real polynomial approximating $f$.  
\end{definition}
We again note easy to observe relations.
\begin{fact}
    \label{fact:adeg-deg-d-relations}
    \(\adeg(f)\leq \deg(f)\leq \dqc(f)\), and\, \(\adeg(f)\leq \rqc(f)\). 
\end{fact}
We also the need the following measure that played the key role in the proof of the sensitivity conjecture \cite{Huang}. 
\begin{definition}
    \label{defi:spectral-sens}
    The \emph{sensitivity graph} of $f$, $G_f=(V,E)$, is the graph where $V=\zone^n$ and $E=\{(x,x^i) \mid f(x)\neq f(x^i),i\in[n],x\in V \}$. Let $A_f$ be the adjacency matrix of the graph $G_f$. Then, the \emph{spectral sensitivity} of $f$, denoted $\lambda(f)$, is the spectral norm  $\|A_f\|$ of the adjacency matrix $A_f$. 
\end{definition}
Note that since $A_f$ is a real symmetric matrix, $\lambda(f)$ is also the maximum eigenvalue of $A_f$. We now note some results that will be useful for us later. 
\begin{fact}[\cite{Huang}]
    \label{fact:lambda-s}
    \(\lambda(f) \leq \s(f)\). 
\end{fact}
\begin{lemma}[\cite{Huang}]
    \label{lem:sens-thm}
    \(\deg(f) \leq \lambda(f)^2\). 
\end{lemma}
\begin{lemma}[\cite{ABK+}]
    \label{lem:lambda-adeg}
    \(\lambda(f) = O(\adeg(f))\). 
\end{lemma}
\begin{lemma}[\cite{ABK+}]
    \label{lem:sens-lambda}
    \(\s(f)\leq \lambda(f)^2\). 
\end{lemma}
\begin{lemma}[\cite{Midri04}]
    \label{lem:midrijanis}
    \(\dqc(f)\leq \deg(f)\bs(f)\). 
\end{lemma}
\begin{lemma}[\cite{DHT17}]
    \label{lem:fbs-ucmin}
    \(\fbs(f)\leq 2\ucmin(f)-1\). 
\end{lemma}
\begin{lemma}[\cite{nisan1989crew}]
    \label{lem:nisan-bs-r}
    \(\bs(f) \leq 3\rqc(f)\). 
\end{lemma}
\begin{lemma}[\cite{NS94}]
    \label{lem:bs-adeg}
    \(\bs(f) = O(\adeg(f)^2)\). 
\end{lemma}
\begin{lemma}[\cite{BBCMdW01}]
    \label{lem:adeg-q}
    \(\adeg(f) \leq 2\qqc(f),\) and\, \(\deg(f)\leq 2\qqc_E(f)\). 
\end{lemma}

We also describe the (disjoint) composition of two Boolean functions.  
\begin{definition}
\label{defi:composition}
     For any two Boolean functions $f\colon\boolfn{n}$ and $g\colon\boolfn{m}$, define the \emph{composed} function $f\circ g\colon \boolfn{nm}$ as follows 
\begin{align*}
    f \circ g (x_{11},\ldots ,x_{1m},\ldots\ldots ,x_{n1},\ldots ,x_{nm}) = f(g(x_1),\ldots ,g(x_n)),
\end{align*}
where $x_i = (x_{i1},\ldots ,x_{im}) \in \zone^m$ for $i \in [n]$. 
\end{definition}
\begin{definition}[Rubinstein’s function \cite{Rub95}]
\label{defi:Rubinstein-function}
Let $g\colon\boolfn{k}$ be such that $g(x)=1$ iff $x$ contains two consecutive ones and the rest of the bits are zeroes. The Rubinstein’s function, denoted $\mathsf{RUB}\colon\boolfn{k^2}$, is defined to be $\mathsf{RUB} = \OR_k \circ g$.
\end{definition}
We also define a modified version of the Rubinstein function which will be 
the witnessing function for our negative results. 
\begin{definition}[Modified Rubinstein function]
\label{defi:modified-rubinstein}
    Define $g\colon\boolfn{k}$ to be $1$ iff the input contains $\sqrt{k}$ consecutive ones and the rest of the bits are zeroes. Then, define $f\colon\boolfn{k^2}$ to be the function $(\OR_k \circ g)(x)$ where $x \in \zone^{k^2}$.  
\end{definition}



\section{Negative Results I: Incondensability of Block Sensitivity and Certificate Complexity}
\label{sec:negative-I}
In this section we prove \cref{thm:bs-cert-condensation-intro}, i.e., block sensitivity, fractional block sensitivity and certificate complexity does not condense losslessly. 
\bsCondensation*
In fact, we will use the same function, Modified Rubinstein function (\cref{defi:modified-rubinstein}), for all three measures. More formally, we establish the following bounds for its complexity measures. 
\begin{theorem}
    \label{thm:bs-cert-condensation}
    There exists a Boolean function $f\colon\boolfn{k^2}$ with 
    $\bs(f)=\fbs(f)=\cert(f)=k^{1.5}$, such that for every partial assignment $\rho\colon[k^2]\to\zonep$ with $|\rho^{-1}(*)|= O(k^{1.5})$ we have $\bs(f|_{\rho}) \leq \fbs(f|_{\rho}) \leq \cert(f|_{\rho}) = O(k)$. 
\end{theorem}
\begin{proof}
     Since for every $f$ and input $x$, $\bs(f,x)\leq \fbs(f,x)\leq \cert(f,x)$, it suffices to show the following two properties hold,  
   \begin{enumerate}[label=\roman*)]
        \item\label{prop1} $\bs(f)=\cert(f)=k^{1.5}$, and 
        \item\label{prop2} for every partial assignment $\rho\colon[k^2]\to\zonep$ with $|\rho^{-1}(*)|=O(k^{1.5})$ we have $\cert(f|_{\rho})=O(k)$. 
    \end{enumerate}
    Let $f$ be the modified Rubinstein function as given in \cref{defi:modified-rubinstein}. Then, \cref{prop1} follows from \cref{lem:bounds-mod-rub} and \cref{prop2} is shown in \cref{lem:bounds-restricted}. Thus, the theorem follows. 
\end{proof}
We now prove the two lemmas to complete the proof.
\begin{lemma}
\label{lem:bounds-mod-rub}
   Let $f$ be the function defined in \cref{defi:modified-rubinstein}. Then, \(\cert_0(f) = k^{1.5},\,  \bs_0(f)= k^{1.5},\, \bs_1(f)=k, \text{ and }\, \cert_1(f) = k\). 
\end{lemma}
\begin{proof}
    Recall the $k^2$ variate function $f = \OR_k\circ g$, where $g\colon\boolfn{k}$ evaluates to $1$ iff the input to $g$ contains $\sqrt{k}$ consecutive ones and the rest of the bits are zeroes. Let $x=(x_{11},x_{12},\ldots ,x_{1k},x_{21},\ldots ,x_{2k},\ldots ,x_{k1},\ldots ,x_{kk})$ be an input to $f$. For $1\leq i\leq k$, the $i$-th copy of $g$ is defined over the set of variables $\{x_{i1},\ldots ,x_{ik}\}$. 
        
    Clearly, $\cert_1(f) = \cert_1(g) \leq k$ and $\bs_1(f) =\bs_1(g) \geq k$, which in turn implies that both are equal to $k$. Note a $0$-certificate for $f$ is a collection of $0$-certificates, one for each copy of $g$. Thus, $\cert_0(f) \leq k\cdot \cert_0(g)$. We now upper bound $\cert_0(g)$. A $0$-input $y\in\zone^k$ of $g$ falls under (at least) one of the following types:
    \begin{enumerate}[label={[\arabic*]}]
        \item \label{0input-1} It contains two $1$-bits that are at least $\sqrt{k}$ apart. That is, there exists $i,j \in [k]$ such that $y_i=y_j=1$ and $|i-j| \geq \sqrt{k}$. For such $0$-inputs, the certificate complexity is $2$. 
        \item \label{0input-2} It contains three input bits $y_i$, $y_j$, $y_\ell$ such that $i<j<\ell$ and $y_i=1, y_j=0$, and $y_\ell=1$. For such $0$-inputs, the certificate complexity is at most $3$. 
        \item \label{0input-3} It either $(\textsc{i})$ contains the substring $01^\ell0$, or $(\textsc{ii})$ equals to $0^{k-\ell}1^\ell$ or $1^\ell0^{k-\ell}$ for some $\ell$ such that $1\leq \ell< \sqrt{k}$. For $0$-inputs of kind $(\textsc{i})$, the certificate complexity is at most $3$, namely the first two bits $(01)$ and the last bit $(0)$ of the substring suffices as a $0$-certificate. 
        For $0$-inputs of kind $(\textsc{ii})$, just the two bits on the boundary suffices, $y_{k-\ell}=0$ and $y_{k-\ell+1}=1$ (or, $y_\ell =1$ and $y_{\ell+1}=0$). Therefore, the certificate complexity is again at most $3$.   
        \item \label{0input-4} The input $y=0^k$. The certificate complexity for this input is $\sqrt{k}$. 
    \end{enumerate}
    We therefore have $\cert_0(f)\leq k\cdot \cert_0(g) \leq k^{1.5}$. 
    
    It is now sufficient to show that $\bs_0(f) \geq k^{1.5}$. Consider the input $0^{k^2}$. It is easily seen that $\bs(f,0^{k^2}) = k\cdot \bs(g,0^k) = k\sqrt{k}$. 
\end{proof}
\begin{lemma}
\label{lem:bounds-restricted} 
    Let $f$ be the function defined in \cref{defi:modified-rubinstein}. Then, for every partial assignment $\rho\colon[k^2]\to\zonep$ 
    with $|\rho^{-1}(\ast)|=O(k^{1.5})$, we have $\cert(f|_{\rho})=O(k)$. 
\end{lemma}
\begin{proof} 
    Recall $f = \OR_k\circ g$, where $g\colon\boolfn{k}$ evaluates to $1$ iff the input to $g$ contains $\sqrt{k}$ consecutive ones and the rest of the bits are zeroes. Let $x=(x_{11},\ldots ,x_{1k},\ldots ,x_{k1},\ldots ,x_{kk})$ be an input to $f$. For $1\leq i\leq k$, let $g_i = g(x_{i1},\ldots ,x_{ik})$ be the $i$-th copy of $g$.

    Since $\cert_1(f) = k$, it suffices to show $\cert_0(f|_\rho) = O(k)$ 
    for every $\rho\colon[k^2]\to\zonep$ with $|\rho^{-1}(\ast)|=O(k^{1.5})$. 

    A $0$-certificate for $f|_\rho$ comprises of at most $k$ $0$-certificates, one for each copy of $g|_\rho$ that remains non-zero. 
    That is, $\cert_0(f|_\rho) \leq \sum_{i=1}^k \cert_0(g_i|_\rho)$. We therefore now bound $\cert_0(g_i|_\rho)$ when $g_i|_\rho$ is not equal to the constant function $0$. We assume, wlog, that $g_1|_\rho$ is 
    non-constant. Based on the restriction $\rho$, we have two cases to consider. 
    \begin{enumerate}
        \item[] \textbf{Case 1}$\colon$ $\rho$ sets some variable of $g_1$ to $1$. That is, there exists $1\leq j\leq k$, such that $\rho(x_{1j}) = 1$.  In this case, $\cert_0(g_1|_\rho) \leq 3$. 

        This is essentially because we only need to certify that the sequence of $1$s in a $0$-input to $g_1|_\rho$ does not satisfy the required property. In particular, any $0$-input of $g_1|_\rho$ is of type (\ref{0input-1}), (\ref{0input-2}), or (\ref{0input-3}). 

        \item[] \textbf{Case 2}$\colon$ $\rho$ sets no variable of $g_1$ to $1$. That is, $\rho$ may set some variables in $\{x_{11},\ldots ,x_{1k}\}$ to $0$ and leaves the rest unset. Let $\eta_1$ be the number of variables of $g_1|_\rho$. That is, $|\rho^{-1}(\ast)\cap\{x_{11},\ldots ,x_{1k}\}| = \eta_1$.

        In such a case, observe that the all-zero input to $g_1|_\rho$ remains the hardest to certify.  Thus, $\cert_0(g_1|_\rho) = \cert(g_1|_\rho,0^{\eta_1}) \leq \eta_1/\sqrt{k}$. The inequality follows because there is a certificate that reveals a $0$-bit for every consecutive $\sqrt{k}$ positions and, hence its size is at most $\eta_1/\sqrt{k}$.  
    \end{enumerate}
    Let $\eta_i = |\rho^{-1}(\ast)\cap\{x_{i1},\ldots ,x_{ik}\}|$ be the number of variables of $g_i|_\rho$ for $1\leq i \leq k$.
    We are now ready to bound $\cert_0(f|_\rho)$. Let $t$ be the number of $g_i$'s such that $g_i|_\rho$ falls in \textbf{Case}~$\mathbf{2}$. Further assume, wlog, that these are $\{g_1, g_2, \ldots ,g_t\}$. Then, we have
    \begin{align*}
        \cert_0(f|_\rho) & \leq \sum_{i=1}^k \cert_0(g_i|_\rho) \leq 3(k-t) +\sum_{i=1}^t\cert_0(g_i|_\rho) \leq 3(k-t) +\sum_{i=1}^t \left(\eta_i/\sqrt{k}\right)\\
        & \leq 3(k-t) + |\rho^{-1}(\ast)|/\sqrt{k} = O(k).
    \end{align*}
\end{proof}

\begin{remark}
We note that the modified Rubinstein function, even with different parameters, cannot give a better counter-example for lossless condensation than \cref{thm:bs-cert-condensation}. We prove this in \cref{thm:optimality-rub} in Appendix~\ref{sec:apx-rub}.  
\end{remark}

\section{Negative Results II: Incondensability of Decision Tree Complexity}
\label{sec:negative-II}
In this section we prove \cref{thm:query-condensation-intro}, i.e., query complexity and its variants does not condense losslessly. In particular, we show that $\dqcs{0}$-query and $\dqcs{\wedge}$-query complexities does not condense losslessly and also present an improved example for the fact that decision tree query complexity does not condense losslessly. In fact, we will use the same function to establish all three claims. 
\DCondensation*

It was shown in \cite{GoosNR024} that there exists a function $h$ such that for every restriction $\rho$ that fixes all but at most $O(\dqc(h))$ variables, the query complexity $\dqc(h|_\rho)$ of the restricted function $h|_\rho$ is $\widetilde{O}(\dqc(h)^{3/4})$. We give an improved example by exhibiting a function $h'$ such that $\dqc(h'|_\rho) = \widetilde{O}(\dqc(h')^{2/3})$.  Our improved example is a variant of the function considered in \cite{GoosNR024}. In particular, our function is also a cheat-sheet version of Tribes, but we use different arities for the basic Tribes function as well as keep the number of cheat-sheet cells much low. 

Furthermore, we will use the same function to show that $\dqcs{0}$ and $\dqcs{\wedge}$ also does not condense losslessly. 
To present the example we will need the \emph{cheat-sheet} framework of \cite{ABK16} which we introduce now.

\begin{definition}[Cheat-sheet version of a function]
    \label{defi:cheat-sheet}
    Fix $t > 0$. Let $f\colon\boolfn{n}$ be a function such that the \emph{description} of any (minimal) certificate takes at most $m$ bits. We then define 
    $$ \fcs\colon\boolfn{n\log t+ mt\log t} $$
    the cheat-sheet version of $f$ as follows. The input to $\fcs$ is $X =(x_1,\ldots , x_{\log t}, Y_0,\ldots ,Y_{t-1})$ where $x_i \in\zone^n$ are inputs to $\log t$ copies of $f$, and $Y_j \in \zone^{m\log t}$ are cheat-sheet cells.

    Define $\fcs(X)=1$ iff the cheat-sheet cell $Y_\ell$, given by the positive integer $\ell$ corresponding to the binary string $f(x_1)f(x_2)\cdots f(x_{\log t})$, describes $\log t$ certificates, one for each $x_i$ using $m$ bits. That is, it contains a \emph{description} of $f(x_i)$-certificate for each $x_i$.  
\end{definition}

For more details on the cheat-sheet framework, we refer to \cite{ABK16, AKK16}. As alluded earlier our function is also a cheat-sheet version of the Tribes function. We consider the following Tribes function.
\begin{definition}[Tribes function]
    \label{defi:tribes}
    Define $f\colon\boolfn{k^{1.5}}$ to be the function $ (\AND_{k} \circ\OR_{\sqrt{k}})(x)$ where $x\in \zone^{k^{1.5}}$. 
\end{definition}
We note a few well-known properties of the Tribes function. 
\begin{proposition}
    \label{prop:tribes-complexity}
    For the Tribes function $f$ as defined in \cref{defi:tribes}, it is known that $\cert_0(f)=\sqrt{k}$, $\cert_1(f)=k$, $\dqc(f) = k\sqrt{k}$ and $\dqcs{0}(f) \geq k\sqrt{k}-k+1$. 
\end{proposition}
\begin{proof}
     Though the bounds on $\cert_0$, $\cert_1$, and $\dqc$ are well-known, the bound on $\dqcs{0}$ is not often encountered. So for completeness we sketch a proof here.  

     Consider the following (usual) adversary algorithm for the Tribes: On a query by a decision tree algorithm, the adversary responds with $0$ if the $\OR$ containing the queried variable has variables which are not queried yet. Otherwise it's the last variable to be queried in that $\OR$, in such a case the adversary responds with $1$ if there are other $\OR$s whose output is not yet fixed. Finally the last variable in the last $\OR$ is queried, which implies the lower bound.    
\end{proof}

We consider the \emph{cheat-sheet} version $\fcs$ of the Tribes function $f$ (\cref{defi:tribes}) where certificates in the cheat-sheet cells are described as follows. 
A cell $Y_\ell$, where $\ell$ in binary is given by 
$\ell_1\ell_2\cdots \ell_{\log t} \in \zone^{\log t}$, describes a 
$\ell_i$-certificate for $x_i$ with respect to $f$ for $i\in[\log t]$. 
\begin{description}
    \item[Describing $1$-certificates:] A $1$-certificate for $f=\AND_k\circ \OR_{\sqrt{k}}$ contains a $1$-certificate for each copy of $\OR$. For each $1$-certificate the description is given by the label of a variable which is set to $1$. This requires $\log k^{1.5}$ bits per $\OR$. Therefore, the full description of a $1$-certificate for $f$ requires $k(\log k^{1.5})$ bits. 
    \item[Describing $0$-certificates:] A $0$-certificate for $f = \AND_k\circ \OR_{\sqrt{k}}$ contains a $0$-certificate for one of the copies of $\OR$. Hence, we describe a $0$-certificate for $f$ by just writing the label of a copy of $\OR$. Thus the full description of a $0$-certificate for $f$ requires $\log k$ bits. 
\end{description}
Therefore, the cheat-sheet version $\fcs$ of the Tribes function is 
\begin{align}
    \label{eq:cheat-sheet-tribes}
    \fcs\colon\boolfn{k^{1.5}(\log t)+ (t\log t)k(\log k^{1.5})},
\end{align}
where its input is $X =(x_1,\ldots , x_{\log t}, Y_0,\ldots ,Y_{t-1})$, and 
$\fcs(X)=1$ iff the cheat-sheet cell $Y_\ell \in \zone^{(\log t)k(\log k^{1.5})}$, given by the positive integer $\ell$ corresponding to the binary string $f(x_1)f(x_2)\cdots f(x_{\log t})$, describes $f(x_i)$-certificate for each $x_i$ as given by the description above. Note that $0$-certificates are much smaller than $1$-certificates, but we will use the same amount of space, $k(\log k^{1.5})$, for each with the notation that in the case of $0$-certificates first $\log k$ bits describe them. 

We now show that for the choice of $t=k^{1.5}$, both $\dqc(\fcs)$ and $\dqcs{0}(\fcs)$ equals $\widetilde{\Theta}(k^{1.5})$, but for every restriction $\rho$ that fixes all but at most $O(\dqc(\fcs))$ variables, the query complexity $\dqc(\fcs|_\rho)$ (and hence $\dqcs{0}(\fcs|_\rho)$) of the restricted function $\fcs|_\rho$ is $\widetilde{O}(k)$. We now prove each of these assertions one by one. In the following fix $t=k^{1.5}$. Then the number of variables for $\fcs$ is $\Theta(k^{2.5}\log^2 k)$. 

\begin{lemma}
    \label{lem:D-cs-lb}
    $\dqc(\fcs) \geq t=k^{1.5}$ and\, $\dqcs{0}(\fcs) \geq t - t^{2/3}$. 
\end{lemma}
\begin{proof}
    We will give an adversary argument for $\fcs$ based on the adversary for $f$ mentioned in \cref{prop:tribes-complexity}. 

    On a query the adversary for $\fcs$ will respond as follows: If the queried variable is from the address part, i.e., $x_1,\ldots , x_{\log t}$, then the response is in accordance with the adversary of $f$ given in \cref{prop:tribes-complexity}. Otherwise the queried variable is from the cheat-sheet cells, i.e., $Y_0,\ldots , Y_{t-1}$, and then the adversary responds with $0$. 

    Now if an algorithm makes only $t-1$ queries and they are answered according to the adversary defined above, then the following observations hold:
    \begin{itemize}
        \item No $f(x_i)$ has been fixed yet,
        \item there is at least one cheat-sheet cell, say $Y_\ell$, which has not been queried at all, and
        \item the number of $0$-queries made until now is at least $t-1-k$. 
    \end{itemize}
    We can therefore extend $x_i$'s in a way that the address points to the cell $Y_\ell$ and, furthermore, the cell $Y_\ell$ can be set as we wish to give values $0$ or $1$ to contradict the algorithm's answer; thus finishing the adversary argument.
\end{proof}

We now present a claim that will be used to construct a query algorithm for $\fcs$ as well as its restricted version.
\begin{claim} 
\label{claim:discarding-cheatsheets}
Given two distinct indices $\ell, p \in\{0,\ldots ,t-1\}$, at least one of the cells $Y_\ell$ and $Y_p$ can be discarded with $O(\log k)$ queries.
\end{claim}
\begin{proof}
    Since $\ell \neq p$, there exists $i\in[\log t]$ such that $\ell_i\neq p_i$. Assume, wlog, $\ell_i=0$ and $p_i=1$. 

    The algorithm queries $Y_\ell$ to know the label of $\OR$ which is a supposed $0$-certificate for $f(x_i)$. Let $\OR_j$, $j\in[k]$, be the copy of $\OR$ within the $i$-th copy of $f$ that $Y_\ell$ has listed as $0$-certificate for $x_i$. Observe that $Y_p$ is supposed to have a $1$-certificate for $\OR_j$. Recall a $1$-certificate for $\OR_j$ is the label of a variable of $\OR_j$.  Query the variables in $Y_p$ and the input $x_i$ to verify the $1$-certificate for $\OR_j$. Now if the claimed $1$-certificate for $\OR_j$ is valid then we can remove $Y_\ell$, else $Y_p$ can be removed. 

    The total number of queries made is at most $O(\log k)$.
\end{proof}
We are now ready to show a nearly tight upper bound on $\dqc(\fcs)$. Using \cref{claim:discarding-cheatsheets}, we prove that the bound of \cref{lem:D-cs-lb} is tight up to logarithmic factors. 
\begin{lemma}
    \label{lem:D-cs-ub}
    $\dqc(\fcs) = O(t\log k ~+~ k\log k\log t)$.
\end{lemma}
\begin{remark}
    We note that the bound in \cref{lem:D-cs-ub} can also be obtained by a straightforward algorithm that solves each copy of $f$ and then verifies the contents of the cheat-sheet cell pointed to by the evaluations of $f$. However, we choose to present an algorithm using \cref{claim:discarding-cheatsheets}, for it will be instructive when designing an algorithm for the restricted function.   
\end{remark}
\begin{proof}
    On an unknown input $X =(x_1,\ldots , x_{\log t}, Y_0,\ldots ,Y_{t-1})$ our query algorithm for $\fcs$ has two stages. At any point in a run of the algorithm, we call a cheat-sheet cell $Y_\ell$ \emph{valid} if the \emph{revealed} contents of $Y_\ell$ are \emph{not} yet found to be inconsistent or invalid with respect to the input $(x_1,x_2,\ldots ,x_{\log t})$ and the function $f$.

    Note that initially the set of valid cells is the set of all cheat-sheet cells $\{Y_\ell \mid 0 \leq \ell \leq t-1\}$.  We now describe the two stages of the algorithm. 
    \begin{description}
        \item[Stage $1$:] In the first stage, we will remove cheat-sheet cells one by one (using \cref{claim:discarding-cheatsheets}) to end up with only one valid cell. This will take $(t-1) \cdot O(\log k)$ many queries. At the end of this stage, suppose the only remaining cheat-sheet cell is $Y_z$. 
        \item[Stage $2$:] In the second stage, we verify the certificates described in $Y_z$ by querying the variables in $Y_z$ and the corresponding certificate variables in $(x_1,\ldots ,x_{\log t})$. If the verification passes, the algorithm outputs $1$, else outputs $0$. Since the algorithm needs to verify the contents of one cell, the total number of queries made in this stage is $O(k\log k\log t + \cert(f) \log t)=O(k\log k\log t)$.  
    \end{description}
    The total number of queries made is the sum of the numbers in each stage, and hence at most $O(t\log k ~+~ k\log k \log t)$. 

    Furthermore, the correctness of the algorithm follows from the fact that a cheat-sheet cell is removed from the set of valid cells if and only if the cell is found to be inconsistent or invalid with respect to input $(x_1,\ldots ,x_{\log t})$ and $f$. The only remaining valid cell is verified in the second stage.
\end{proof}

We now finish the last, but the most important, task of showing that indeed the decision tree complexity of the restricted function reduces; thereby finishing the counterexample for the hardness condensation of decision tree complexity and its variants. 
\begin{lemma}
    \label{lem:D-cs-restricted-ub}
    For every partial assignment $\rho$ to the variables of $\fcs$ with $|\rho^{-1}(\ast)| = O(\dqc(\fcs)) = O(k^{1.5}\log k)$, we have $\dqc(\fcs|_\rho) = O(k\log^2 k)$. 
\end{lemma}
\begin{proof}
    Fix a partial assignment $\rho$ with $|\rho^{-1}(\ast)|=O(\dqc(\fcs))=O(k^{1.5}\log k)$. Our query algorithm for $\fcs|_\rho$ on an unknown input $X|_\rho =(x_1|_\rho,\ldots , x_{\log t}|_\rho, Y_0|_\rho,\ldots ,Y_{t-1}|_\rho)$ has two parts like the unrestricted case. 

    In the first part we identify a cheat-sheet cell (to be verified later) by zeroing in on the string $z\in\zone^{\log t}$ that the address part $f(x_1|_\rho)f(x_2|_\rho)\cdots f(x_{\log t}|_\rho)$ can take if the restricted function $\fcs|_\rho$ were to evaluate to $1$. 

    In the second part, like the unrestricted case, we verify the certificates described in $Y_z$ by querying the variables in $Y_z|_\rho$ and the corresponding certificate variables in $(x_1|_\rho,\ldots ,x_{\log t}|_\rho)$. If the verification passes, the algorithm outputs $1$, else outputs $0$. Since the algorithm needs to verify the contents of one cell, the total number of queries made in the second part is $O(k\log k\log t + \cert(f) \log t)=O(k\log^2 k)$.
    
    \paragraph*{Algorithm for the first part} We now give the details of our algorithm for the first part. Recall that the address part of the input contains $\log t$ copies of the base function $f$. We will run (at most) $\log t$ many iterations of the algorithm. At the end of each iteration $i$, we will have a candidate value for $f(x_i|_\rho)$, i.e., either we know the value of $f(x_i|_\rho)$ or there are no valid cheat-sheet cells for $f(x_i|_\rho)=1$. 
    At the end of the $\log t$ iterations, there will only be one remaining (possibly valid) cheat-sheet cell to be verified. We now give details of an iteration of the algorithm. 

    Let $c$ be a constant such that $|\rho^{-1}(\ast)|=O(\dqc(\fcs)) \leq ck^{1.5}\log k$.

    \noindent{\bf Details of a particular iteration:} Fix an $i\in [\log t]$. At any given point of the iteration let $\mc{V}$ be the set of ``alive'' cheat-sheets (cheat-sheets which have not been found inconsistent) and $\mc{V}_1$ be the cheat-sheet cells in $\mc{V}$ where, supposedly, $f(x_i|_\rho)=1$, (i.e., the $i$-th bit of the index of the cell is $1$). 
    On the other hand, let $\mc{O}_0$ denote the set of $\OR_j$'s ($j \in [k]$) in the $i$-th copy of $f$ which could still possibly be $0$. Initially, $\abs{\mc{O}_0} = k$ and, in the first iteration, $|\mc{V}|= t$ and $|\mc{V}_1|= t/2$. Also, note that each cell in $\mc{V}_1$ is supposed to contain a $1$-certificate for $k$ many $\OR$s.

    In each iteration the algorithm will have two stages; where the first stage (and arguably the more involved one) will get us to a situation which can be handled by previous techniques. More formally, after the first stage, either $\abs{\mc{O}_0} \leq c\sqrt{k}\log k$ or $\abs{\mc{V}_1} \leq 2k$. 

    The second stage takes care of these two cases. In the first case we can check all the alive $\OR$'s exhaustively. For the second case, using \cref{claim:discarding-cheatsheets} we will get a single cheat-sheet cell in $\abs{\mc{V}_1}$ which is then verified to ascertain whether indeed it contains a $1$-certificate for $f(x_i|_\rho)$.

    We now describe the two stages in detail. 
    \begin{description}
        \item[Stage $1$:] We can assume that $\abs{\mc{O}_0} > c\sqrt{k}\log k$ and $\abs{\mc{V}_1} > 2k$, otherwise we move to the second stage.
        
        Looking at $\mc{V}_1$, for each $\OR_j \in \mc{O}_0$, we call an $\OR_j$ to be \emph{revealed} if (supposed) $1$-certificates for that particular $\OR_j$ is completely fixed by $\rho$ (i.e., not a single $\ast$) in at least $\abs{\mc{V}_1}/2$ many cells in $\mc{V}_1$. 
        Notice that if there are no revealed $\OR$'s, then for more than $c\sqrt{k}\log k$ many $\OR$'s there are at least $\abs{\mc{V}_1}/2 > k$ many $\ast$'s. This is a contradiction because the number of $\ast$'s, by our assumption, is at most $ck^{1.5}\log k$.

        For a revealed $j$, there exists a \emph{popular} variable, say $l \in [\sqrt{k}]$, which appears in at least $\abs{\mc{V}_1}/2\sqrt{k}$ many cells in $\mc{V}_1$. We query $x_{i,j,l}$ from the input part $x_i$. If it is $1$ then $\OR_j$ can be discarded from $\mc{O}_0$, otherwise it is $0$ and at least $\abs{\mc{V}_1}/2\sqrt{k}$ many cells in $\mc{V}_1$ can be thrown out. 

        Essentially, the first stage of the iteration can be described as: find a revealed $\OR$ and query the popular variable till $\abs{\mc{O}_0} \leq c\sqrt{k}\log k$ or $\abs{\mc{V}_1} \leq 2k$. The existence of such a revealed $\OR$ follows from the argument described above. We now bound the query complexity in this stage.

        Notice that after every query either $\abs{\mc{O}_0}$ reduces by $1$ (popular variable being $1$) or $\abs{\mc{V}_1}$ reduces by a factor of $(1-\frac{1}{2\sqrt{k}})$ (popular variable being $0$). So the query complexity of the first stage is bounded by $O(k + \sqrt{k} \log t) = O(k)$.

        \item[Stage $2$:] We know that either $\abs{\mc{O}_0} \leq c\sqrt{k}\log k$ or $\abs{\mc{V}_1} \leq 2k$. As alluded before, for the first case, we can query all the $\OR$'s in $\mc{O}_0$ and find the value of $f(x_i|_\rho)$. Since the arity of each $\OR$ is $\sqrt{k}$, this requires $O(k\log k)$ queries. 

        For the second case, use \cref{claim:discarding-cheatsheets} to reduce the size of $\mc{V}_1$ to one using $O(k\log k)$ many queries. Then, in the remaining cell, the claimed $1$-certificate for $f(x_i|_\rho)$ can be verified by reading the entire certificate in $O(k \log k)$ queries. In other words, the second stage makes $O(k\log k)$ many queries.         
    \end{description}
    Hence, the total number of queries made in both stages of the iteration is $O(k\log k)$. This finishes one iteration $i \in [\log t]$ of the algorithm. 
    \begin{description}
        \item[Query complexity of the entire algorithm:] Since there are $\log t$ many iterations, the total queries made in the first part of the algorithm is $O(k\log^2 k)$. After these $\log t$ iterations, in the second part, there is at most one possible valid cheat-sheet, that can be checked in $O(k\log^2 k)$ many queries (as described in the beginning). Therefore, the entire algorithm makes at most $O(k\log^2 k)$ queries.
        \item[Correctness:] The correctness of the algorithm follows from the fact that a cheat-sheet cell is removed from the set of valid cells if and only if the cell is found to be inconsistent with respect to input $(x_1|_\rho,\ldots ,x_{\log t}|_\rho)$ and $f$, or when we know the value of $f(x_i)$. Further, the only remaining valid cell is verified in the second part.
    \end{description}
\end{proof}

As a corollary to \cref{lem:D-cs-lb}, \cref{lem:D-cs-ub}, \cref{lem:D-cs-restricted-ub} and \cref{fact:zerodt-anddt-dt}, we obtain an improved example for the incondensability of query complexity and as well as show incondensability of $\dqcs{0}$ and $\dqcs{\wedge}$.  
\begin{theorem}
    \label{thm:query-condensation}
    Let $\fcs$ be the Boolean function on $\Theta(k^{2.5}\log^2 k)$ variables as defined in \cref{eq:cheat-sheet-tribes}. Then, the following holds
    \begin{itemize}
        \item $\Omega(k^{1.5})=\dqcs{0}(\fcs) \leq \dqcs{\wedge}(\fcs) \leq  \dqc(\fcs) = O(k^{1.5}\log k)$, and 
        \item for all restrictions $\rho$ with $|\rho^{-1}(\ast)| = O(\dqc(\fcs))$, we have $\dqcs{0}(\fcs|_\rho)\leq\dqcs{\wedge}(\fcs|_\rho)\leq\dqc(\fcs|_\rho)=O(k\log^2 k)=\widetilde{O}(\dqcs{0}(\fcs)^{2/3})$. 
    \end{itemize}
\end{theorem}

This completes the proof of \autoref{thm:query-condensation-intro}. We take a moment to note that if we consider the function $h = \OR_k\circ\AND_{\sqrt{k}}$ and the cheat-sheet version $h_{\textup{CS}}$ of $h$, a dually similar lower bound and algorithm for $h_{\textup{CS}}$ and its restricted version exists, which implies the following theorem on the incondensability of $\OR$-decision tree complexity. 
\DCondensationOR*

We further observe that the zero-error randomized query complexity $\rqc_0(\fcs)$ of $\fcs$ is also high, i.e., $\Omega(k^{1.5})$, thus obtaining an improved example for the incondensability of $\rqc_0$ (\cref{thm:R0-condensation-intro}).  
\begin{theorem}
    \label{thm:Rquery-condensation}
    Let $\fcs$ be the Boolean function on $\Theta(k^{2.5}\log^2 k)$ variables as defined in \cref{eq:cheat-sheet-tribes}. Then, the following holds
    \begin{itemize}
        \item $\rqc_0(\fcs) = \Tilde{\Theta}(k^{1.5})$, and 
        \item for all restrictions $\rho$ with $|\rho^{-1}(\ast)| = O(\rqc_0(\fcs))$, we have $\rqc_0(\fcs|_\rho) = O(k\log^2 k) = \widetilde{O}(\rqc_0(\fcs)^{2/3})$. 
    \end{itemize}
\end{theorem}
\begin{proof}
    Since $\rqc_0 \leq \dqc$, it suffices to show that $\rqc_0(\fcs) = \Omega(k^{1.5})$ (\cref{thm:query-condensation} will then complete the proof).
    To prove the lower bound on $\rqc_0(\fcs)$, note that $\rqc_0(f) = \Omega(\rqc(f))$ and $\rqc(f) = \Theta(k^{1.5})$ \cite{JK10}. See also \cite{GJPW18, AKK16, CKMPSS23}.  
    We will show a reduction from $\fcs$ to $f$.  
    
    Given an input $x$ for $f$, the input of $\fcs$ (say $y$) will consist of $\log t$ copies of $x$ in the address part and all $0$'s in the cheat-sheet part. Whenever the $\rqc_0$ algorithm for $\fcs$ answers correctly (i.e., does not say ``don't know''), it has queried a certificate for $y$ (say $C_y$). Note that the number of cheat-sheets are more than the allowed number of queries to the $\rqc_0$ algorithm. This implies that the certificate $C_y$ hasn't queried at least one cheat-sheet cell at all.
    
    We claim that this certificate $C_y$ (that hasn't queried any 
    element of a cheat-sheet cell) has a certificate for at least one address bit. To prove the claim by contradiction, if $C_y$ does not have a certificate for $x$ (i.e., does not certify any of the address bits), the address bits of $y$ could be flipped (keeping it consistent with $C_y$) to point to the cheat-sheet cell which has not been queried. Since $C_y$ does not contain any element of this cheatsheet cell but an input consistent with $C_y$ could point to this cheatsheet cell, we have the contradiction.

   To summarize, the $\rqc_0$ algorithm for $\fcs$ gives a certificate for $y$ which contains a certificate for at least one address bit. Since the certificate for any address bit is a certificate for $x$, this will finish the $\rqc_0$ algorithm for $f$.
\end{proof}

\section{Positive Results: Lossy Condensation}
\label{sec:positive-results}
It is a natural question to ask if the gaps in \cref{thm:bs-cert-condensation-intro} and \cref{thm:query-condensation-intro} are optimal. In other words, does there exists a function for which hardness is even less condensible than shown in these theorems?

In particular, for deterministic query complexity, it asks if there exists a function $f$ such that for all restrictions $\rho$ with $|\rho^{-1}(\ast)|=O(\dqc(f))$, we have $\dqc(f|_\rho) = O(\dqc(f)^\alpha)$ where $\alpha < 2/3$? 
Observe that $\alpha \geq 1/3$, since $\dqc(f)\leq\deg(f)^3$ \cite{Midri04} and $\deg(f)$ condenses losslessly.  Note that if the improved bound $\dqc(f)\leq \deg(f)^2$ were known to hold then we will have a better lower bound for the exponent. 

Similarly, for block sensitivity, the exponent $\alpha \geq 1/4$, since $\bs(f)\leq \s(f)^4$ \cite{Huang} and $\s(f)$ condenses losslessly. Note again that if the improved bound $\bs(f)\leq \s(f)^2$ were known to hold then we will have a better lower bound for the exponent. 

In this section we will establish an improved lower bound for the exponent for almost all decision tree based complexity measures. 
\PositiveResults*

We break the proof of \cref{thm:exponent-lb} in multiple propositions for the sake of readability. \cref{thm:exponent-lb}~(\ref{thm4:a}) is proved in \cref{prop:bs-cert-exponent-lb}, \cref{prop:query-exponent-lb}, and \cref{prop:adeg-exponent-lb}, while \cref{thm:exponent-lb}~(\ref{thm4:b}) is proved in \cref{prop:rquery-exponent-lb} and \cref{prop:exact-qq-exponent-lb}. Finally, \cref{thm:exponent-lb}~(\ref{thm4:c}) is proved in \cref{prop:qq-exponent-lb}.

Note that we can relax the requirement of $|\rho^{-1}(\ast)|=\Theta(\mc{M}(f))$ to $|\rho^{-1}(\ast)| = O(\mc{M}(f))$. This is because all the measures mentioned in \cref{thm:exponent-lb} are non-increasing under restrictions. 

\begin{proposition}
    \label{prop:bs-cert-exponent-lb}
    For every Boolean function $f\colon\boolfn{n}$, there exists a restriction $\rho\in\zonep^{n}$ with $|\rho^{-1}(\ast)| = O(\mc{M}(f))$ such that $\mc{M}(f|_\rho) = \Omega(\sqrt{\mc{M}(f)})$, where $\mc{M}\in\{\bs,\fbs,\cert\}$.  
\end{proposition}
\begin{proof}
    We have two cases based on whether or not $\s(f)\geq\sqrt{\mc{M}(f)}$. 
    \begin{description}
        \item[Case $1$:] $\s(f)\geq\sqrt{\mc{M}(f)}$. Since sensitivity condenses losslessly and $\s(f)\leq\mc{M}(f)$ (\cref{lem:s-bs-c-relation}), we have a restriction $\rho$ with $|\rho^{-1}(\ast)| = O(\mc{M}(f))$ and $\mc{M}(f|_\rho)\geq\s(f|_\rho)\geq\sqrt{\mc{M}(f)}$. 
        \item[Case $2$:] $\s(f) < \sqrt{\mc{M}(f)}$. It is known that $\mc{M}(f)\leq\bs(f)\s(f)$ (\cref{lem:s-bs-c-relation}). We thus have $\bs(f) > \sqrt{\mc{M}(f)}$. 
        Let $x$ be an input to $f$ such that $\bs(f,x)=\bs(f) > \sqrt{\mc{M}(f)}$. Let $B_1,B_2,\ldots ,B_{\sqrt{\mc{M}(f)}}$ be a set of disjoint minimal sensitive blocks at $x$. Consider a restriction $\rho$ that fixes variables \textbf{not} in $\cup_{i=1}^{\sqrt{\mc{M}(f)}} B_i$ according to $x$ and leaves the variables in $\cup_{i=1}^{\sqrt{\mc{M}(f)}} B_i$ unset. Since the size of a minimal sensitive block is at most $\s(f)$ (\cref{lem:minimal-sens-block}), we clearly have $|\rho^{-1}(\ast)| = O(\mc{M}(f))$. Furthermore, $\mc{M}(f|_\rho)\geq \bs(f|_\rho)\geq\bs(f|_\rho,x)\geq\sqrt{\mc{M}(f)}$. \qedhere
    \end{description}
\end{proof}

We now establish the lower bound for the deterministic query complexity and unambiguous certificate complexities. 
\begin{proposition}
    \label{prop:query-exponent-lb}
    For every Boolean function $f\colon\boolfn{n}$, there exists a restriction $\rho\in\zonep^{n}$ with $|\rho^{-1}(\ast)| = O(\mc{M}(f))$ such that $\mc{M}(f|_\rho) = \Omega(\sqrt{\mc{M}(f)})$, where $\mc{M}\in\{\ucmin,\uc_1, \uc, \dqc\}$. 
\end{proposition}
\begin{proof}
    We have three cases based on whether $\s(f)$ or $\deg(f)\geq\sqrt{\mc{M}(f)}$.
    \begin{description}
        \item[Case $1$:] $\s(f)\geq\sqrt{\mc{M}(f)}$. Since sensitivity condenses losslessly and $\s(f)\leq\mathcal{M}(f)$ (\cref{fact:uc-d-relation} or \cref{lem:fbs-ucmin} depending on $\mc{M}$), we have a restriction $\rho$ with $|\rho^{-1}(\ast)| = O(\mc{M}(f))$ and $\mc{M}(f|_\rho)\geq\s(f|_\rho)\geq\sqrt{\mc{M}(f)}$. 
        \item[Case $2$:] $\deg(f)\geq\sqrt{\mc{M}(f)}$. Since degree condenses losslessly and $\deg(f)\leq\mc{M}(f)$ (\cref{fact:uc-d-relation}), we have a restriction $\rho$ with $|\rho^{-1}(\ast)| = O(\mc{M}(f))$ and $\mc{M}(f|_\rho)\geq\deg(f|_\rho)\geq\sqrt{\mc{M}(f)}$.
        \item[Case $3$:] $\s(f)< \sqrt{\mc{M}(f)}$ and $\deg(f)< \sqrt{\mc{M}(f)}$. It is known that $\mc{M}(f)\leq\deg(f)\bs(f)$ (\cref{lem:midrijanis}). We thus have $\bs(f)> \sqrt{\mc{M}(f)}$. 
        Let $x$ be an input to $f$ such that $\bs(f,x)=\bs(f)> \sqrt{\mc{M}(f)}$. Let $B_1,B_2,\ldots ,B_{\sqrt{\mc{M}(f)}}$ be a set of disjoint minimal sensitive blocks at $x$. Consider a restriction $\rho$ that fixes variables \textbf{not} in $\cup_{i=1}^{\sqrt{\mc{M}(f)}} B_i$ according to $x$ and leaves the variables in $\cup_{i=1}^{\sqrt{\mc{M}(f)}} B_i$ unset. Since the size of a minimal sensitive block is at most $\s(f)$ (\cref{lem:minimal-sens-block}), we clearly have $|\rho^{-1}(\ast)| = O(\mc{M}(f))$. Furthermore, $\mc{M}(f|_\rho)\geq \bs(f|_\rho)\geq\bs(f|_\rho,x)\geq\sqrt{\mc{M}(f)}$, where the first inequality uses \cref{fact:uc-d-relation} or \cref{lem:fbs-ucmin}. \qedhere 
    \end{description} 
\end{proof}

A similar argument using block sensitivity instead of degree gives lower bound in the case of randomized query complexity.
\begin{proposition}
\label{prop:rquery-exponent-lb}
    For every Boolean function $f\colon\boolfn{n}$, there exists a restriction $\rho\in\zonep^{n}$ with $|\rho^{-1}(\ast)| = O(\mc{M}(f))$ such that $\mc{M}(f|_\rho) = \Omega({\mc{M}(f)}^{1/3})$, where $\mc{M}\in\{\rqc, \rqc_0\}$.
\end{proposition}
\begin{proof}
    We have two cases based on whether or not $\s(f)\geq {\mc{M}(f)}^{1/3}$. Define $t:={\mc{M}(f)}^{1/3}$. 
    \begin{description}
        \item[Case $1$:] $\s(f) \geq t$. 
        Since sensitivity condenses losslessly and $\s(f)= O(\mc{M}(f))$ (\cref{lem:nisan-bs-r}), we have a restriction $\rho$ with $|\rho^{-1}(\ast)| = O(\mc{M}(f))$ and $\mc{M}(f|_\rho)=\Omega(\s(f|_\rho))= \Omega(t)=\Omega({\mc{M}(f)}^{1/3})$, where the first inequality uses \cref{lem:nisan-bs-r}. 
        \item[Case $2$:] $\s(f)< t$. It is known that $\mc{M}(f)\leq\deg(f) \bs(f)$ (\cref{lem:midrijanis}) and $\deg(f)\leq\bs(f)^2$ (\cref{lem:sens-thm}). We thus have $\bs(f)\geq t$. Let $x$ be an input to $f$ such that $\bs(f,x)=\bs(f)\geq t$. Let $B_1,B_2,\ldots ,B_t$ be a set of disjoint minimal sensitive blocks at $x$. 
        Consider a restriction $\rho$ that fixes variables \textbf{not} in 
        $\cup_{i=1}^{t} B_i$ according to $x$ and leaves the variables in 
        $\cup_{i=1}^{t} B_i$ unset. Since the size of a minimal sensitive block is at most $\s(f)$ (\cref{lem:minimal-sens-block}), we clearly have $|\rho^{-1}(\ast)| = O(t^2)= O(\mc{M}(f))$. Furthermore, 
        $\mc{M}(f|_\rho)=\Omega(\bs(f|_\rho))=\Omega(\bs(f|_\rho,x))= \Omega(t)=\Omega({\mc{M}(f)}^{1/3})$, where the first inequality uses \cref{lem:nisan-bs-r}. \qedhere 
    \end{description}
\end{proof}

We now present a stronger hardness condensation for the measures sensitivity and degree. 
We need this to obtain lower bounds for measures like approximate degree, quantum query complexity, etc. It shows that sensitivity and degree can be condensed to any parameter below itself. 
\begin{lemma}[Stronger hardness condensation]
    \label{lem:s-deg-stronger}
    Let $\mc{M}\in\{\s,\deg\}$ and $f\colon\boolfn{n}$ be a Boolean function. Let $k\leq \mc{M}(f)$. Then, there exists a restriction 
    $\rho\in\zonep^{n}$ with $|\rho^{-1}(\ast)| = k$ 
    such that $\mc{M}(f|_\rho)=k$. 
\end{lemma}
\begin{proof}
    We argue each case separately. 

    \noindent{\bf $\mc{M}=\s$:} Let $a\in\zone^n$ be an input to $f$ such that $\s(f,a)=\s(f)$. Given any integer $k\leq\s(f)$, consider a restriction $\rho$ that leaves any $k$ sensitive variables of $a$ unset and sets the rest of the variables according to $a$. Clearly, $|\rho^{-1}(\ast)| = k$ and $\s(f|_\rho)=k$ as well.

    \noindent{\bf $\mc{M}=\deg$:} Let $k\leq\deg(f)$. Let $p$ be the unique multilinear polynomial that represents $f$. Suppose $p$ contains a monomial $\mathfrak{m}$ of degree $k$. Consider a restriction $\rho$ that leaves variables in $\mathfrak{m}$ unset and sets the rest of the variables to $0$. Clearly, $p|_\rho$ is the unique polynomial of degree $k$ that represents $f|_\rho$ over $k$ variables. 

    Now suppose $p$ does not contain any monomial of degree $k$. Let $d>k$ be the smallest positive integer such that there exists a monomial 
    $\mathfrak{m}$ of degree $d$ in $p$, and there does not exist a monomial of degree $i\in [k, d-1]$ in $p$. Consider a restriction $\rho$ that leaves any $k$ variables in $\mathfrak{m}$ unset, sets the rest of $d-k$ variables in $\mathfrak{m}$ to $1$, and sets the rest of variables not in $\mathfrak{m}$ to $0$. The resulting monomial in the restriction 
    (from $\mathfrak{m}$) will have a non-zero coefficient and cannot be cancelled by any other restricted monomial from the polynomial $p$. So, such a restriction $\rho$ will give us a polynomial of degree $k$ that represents $f|_\rho$ over $k$ variables.
\end{proof}

We are now ready to prove lower bounds for approximate degree, quantum query complexity and spectral sensitivity.
\begin{proposition}
\label{prop:adeg-exponent-lb}
     For every Boolean function $f\colon\boolfn{n}$, there exists a restriction $\rho\in\zonep^{n}$ with $|\rho^{-1}(\ast)|=O(\mc{M}(f))$ such that $\mc{M}(f|_\rho)=\Omega(\sqrt{\mc{M}(f)})$, 
     where $\mc{M}\in\{\adeg,\lambda\}$. 
\end{proposition}
\begin{proof}
    Let $k=\lceil\mc{M}(f)\rceil$ where $\mc{M}\in\{\adeg,\lambda\}$. Then, from \cref{lem:s-deg-stronger}, we have restrictions $\rho_1$ with 
    $|\rho_1^{-1}(\ast)|=k$ such that $\deg(f|_{\rho_1})=k$, and $\rho_2$ with 
    $|\rho_2^{-1}(\ast)|=k$ such that $\s(f|_{\rho_2})=k$. This implies  
    \begin{align*}
        \adeg(f|_{\rho_1}) & =\Omega\left(\sqrt{\deg(f|_{\rho_1})}\right) = \Omega(\sqrt{k})=\Omega\left(\sqrt{\adeg(f)}\right), \text{ and } \\
        \lambda(f|_{\rho_2}) & =\Omega\left(\sqrt{\s(f|_{\rho_2})}\right) =\Omega(\sqrt{k})=\Omega\left(\sqrt{\s(f)}\right).  
    \end{align*}
    The first inequality in each case follows from the relations 
    $\deg(f)=O(\adeg(f)^2)$ (\cref{lem:sens-thm} and \cref{lem:lambda-adeg}) and $\s(f)\leq\lambda(f)^2$ (\cref{lem:sens-lambda}), respectively. 
\end{proof}
\begin{proposition}
    \label{prop:exact-qq-exponent-lb}
    For every Boolean function $f\colon\boolfn{n}$, there exists a restriction $\rho\in\zonep^{n}$ with $|\rho^{-1}(\ast)| = O(\qqc_E(f))$ such that 
    $\qqc_E(f|_\rho) = \Omega({\qqc_E(f)}^{1/3})$. 
\end{proposition}
\begin{proof}
    It follows from \cref{lem:midrijanis} that $\deg(f)\geq{\qqc_E(f)}^{1/3}$. Using \cref{lem:s-deg-stronger}, we get a restriction $\rho$ with 
    $|\rho^{-1}(\ast)|=O(\qqc_{E}(f))$ such that $\deg(f|_{\rho})\geq{\qqc_{E}(f)}^{1/3}$. 
    It then follows from \cref{lem:adeg-q} that 
    $\qqc_E(f|_\rho)=\Omega(\deg(f|_\rho))=\Omega({\qqc_E(f)}^{1/3})$. 
\end{proof}
\begin{proposition}
    \label{prop:qq-exponent-lb}
    For every Boolean function $f\colon\boolfn{n}$, there exists a restriction $\rho\in\zonep^{n}$ with $|\rho^{-1}(\ast)| = O(\qqc(f))$ such that 
    $\qqc(f|_\rho) = \Omega({\qqc(f)}^{1/4})$.  
\end{proposition}
\begin{proof}
    We have three cases based on whether $\s(f)$ or 
    $\deg(f)\geq\sqrt{\qqc(f)}$.
    \begin{description}
        \item[Case $1$:] $\s(f)\geq\sqrt{\qqc(f)}$. Using \cref{lem:s-deg-stronger}, we get a restriction $\rho$ with $|\rho^{-1}(\ast)|=O(\qqc(f))$ such that $\s(f|_{\rho})\geq\sqrt{\qqc(f)}$. It then follows from \cref{lem:adeg-q} and \cref{lem:bs-adeg} that 
        $\qqc(f|_{\rho})=\Omega(\sqrt{\s(f|_{\rho})})=\Omega({\qqc(f)}^{1/4})$. 
        \item[Case $2$:] $\deg(f) \geq \sqrt{\qqc(f)}$. Using \cref{lem:s-deg-stronger}, we get a restriction $\rho$ with $|\rho^{-1}(\ast)|=O(\qqc(f))$ such that $\deg(f|_{\rho})\geq\sqrt{\qqc(f)}$. It then follows from \cref{lem:adeg-q}, \cref{lem:lambda-adeg} and \cref{lem:sens-thm} that $\qqc(f|_{\rho})=\Omega(\sqrt{\deg(f|_{\rho})})=\Omega({\qqc(f)}^{1/4})$. 
        \item[Case $3$:] $\s(f)< \sqrt{\qqc(f)}$ and $\deg(f)< \sqrt{\qqc(f)}$. As before, we must have $\bs(f)>\sqrt{\qqc(f)}$ 
        since it is known that $\qqc(f)\leq\dqc(f)\leq\deg(f)\bs(f)$ (\cref{lem:midrijanis}). 
        Let $x$ be an input to $f$ such that $\bs(f,x)=\bs(f) > \sqrt{\qqc(f)}$. Let $B_1,B_2,\ldots ,B_{\sqrt{\qqc(f)}}$ be a set of disjoint minimal sensitive blocks at $x$. Consider a restriction $\rho$ that fixes variables \textbf{not} in $\cup_{i=1}^{\sqrt{\qqc(f)}} B_i$ according to $x$ and leaves the variables in 
        $\cup_{i=1}^{\sqrt{\qqc(f)}} B_i$ unset. Since the size of a minimal sensitive block is at most $\s(f)$ (\cref{lem:minimal-sens-block}), we clearly have $|\rho^{-1}(\ast)| = O(\qqc(f))$. Furthermore, 
        $\qqc(f|_\rho)=\Omega(\sqrt{\bs(f|_\rho)})=\Omega(\sqrt{\bs(f|_\rho,x)})=\Omega({\qqc(f)}^{1/4})$, where the first inequality follows from \cref{lem:adeg-q} and \cref{lem:bs-adeg}.  \qedhere
    \end{description}
\end{proof}

\section{Conclusion}
\label{sec:conclusion}

Given that hardness condensation has been pretty useful in proving sharper lower bounds in circuit complexity theory and proof complexity, a natural question to ask is: ``how much'' of hardness condensation is possible for the complexity measures related to decision tree complexity? Taking cue from G\"{o}\"{o}s, Newman, Riazanov and Sokolov \cite{GoosNR024}; we initiate the study for almost all established complexity measures and prove both negative and positive results about condensation (see Table~\ref{table:complexity-measures}).

In terms of negative results we show that lossless condensation is not possible for block sensitivity, certificate complexity, $\AND$-decision tree query complexity, and several other complexity measures. Additionally, we provide examples of functions whose hardness is even less condensable than the example given by \cite{GoosNR024} for deterministic query complexity and zero-error randomized query complexity. 

To complement these negative results, we also prove that there exist variable restrictions that yield polynomial condensation gaps for different complexity measures. 
However, it remains open whether there exist functions that exhibit even stronger condensation barriers than what is established in this work (see Table~\ref{table:complexity-measures}).
We list several other open questions for future work. 
\begin{enumerate}
    \item An important unresolved question is whether deterministic communication complexity condenses losslessly. 
    Since the $\AND$-decision tree complexity does not condense (Theorem~\ref{thm:query-condensation-intro}), is it possible to show that the answer is negative (as conjectured by~\cite{GoosNR024}) by lifting with 2-bit $\AND$ gadget?
    \item It is also interesting to study the following version of lossy condensation. For a measure $\mc{M}(f)$, does there exist a restriction $\rho$ with $|\rho^{-1}(\ast)|=O(\mathsf{poly}(\mc{M}(f)))$ such that $\mc{M}(f|_\rho)=\Theta(\mc{M}(f))$? It is easily seen that block sensitivity witnesses such a condensation; however, 
    for certificate complexity this question has close connections to the well-studied problem of kernelization of $d$-hitting set problem (cf. \cite{FLLSTZ23}). 
    \item Another avenue to explore is condensation with respect to projections or other similar restricted reductions. Typically, such modifications are used in applications in proof complexity (see, e.g., \cite{Razborov17,Razborov18,BN20,FPR22,BussThapen24,CD24,dRFJNP24}).
    \item As mentioned in the introduction (\cref{sec:intro}, see also \cite{hart2025condensing}), "cross condensation" between Boolean function parameters is another interesting variant of hardness condensation to explore. For example, if there is a Boolean function $f$ such that for any restriction $\rho$, on $O(s(f))$ many variables $\bs(f_\rho) = \omega(\sqrt{\bs(f)})$, then $f$ must have a super-quadratic gap between sensitivity and block sensitivity. 
\end{enumerate}

\bibliography{ref}
\newpage

\appendix
\section{Modified Rubinstein with different parameters cannot improve the negative result} 
\label{sec:apx-rub}
We have shown in \cref{thm:bs-cert-condensation-intro} that modified Rubinstein function (\cref{defi:modified-rubinstein}) witnesses the fact that $\bs$ condenses to $\bs^{2/3}$ for every restriction leaving $\bs$ many variables unset.  
We now show that our analysis is tight and that the Modified Rubinstein function cannot give an improved counter-example than established in \cref{thm:bs-cert-condensation}. We begin with a generic definition of Rubinstein function. 

\begin{definition}[Parameterized Rubinstein]
    \label{defi:parametrized-rub}
    Define $g:\zone^{k_1k_2} \to \zone$ to be $1$ iff the input contains ${k_1}$ consecutive $1$'s and the rest of the bits are $0$. Then, define $f:\zone^{k_1k_2k_3}\to \zone$ to be the function $(\OR_{k_3}\circ g)(x)$ where $x \in \zone^{k_1k_2k_3}$ and $k_1, k_2$, and $k_3$ are positive integers. 
\end{definition}
We observe some easy to prove properties of the parameterized Rubinstein function. 
\begin{lemma}
    \label{lem:parameterized-rub-bounds}
    Let $g$ and $f$ be as defined in \cref{defi:parametrized-rub}. Then, the following holds
    \begin{enumerate}[label=\alph*)]
        \item \(\s_1(g)=\bs_1(g)=\cert_1(g)=k_1k_2\),
        \item \(\s_0(g)=2,\, \bs_0(g)=\cert_0(g)=k_2\),
        \item \(\s_1(f)=\bs_1(f)=\cert_1(f)=k_1k_2\), 
        \item  \(\s_0(f)=2k_3\), and \(\bs_0(f)=\cert_0(f)=k_2k_3\). 
    \end{enumerate}
\end{lemma}

We are now ready to show that a different parameterization of the Rubinstein function will not give an improved counter-example. 
\begin{theorem}
    \label{thm:optimality-rub}
    Let $f\colon\boolfn{k_1k_2k_3}$ be the Parametrized Rubinstein function defined in \cref{defi:parametrized-rub}. Then, there exists a restriction $\rho\colon[k_1k_2k_3]\to\zonep$ with $|\rho^{-1}(\ast)| \leq \bs(f)=\cert(f)$  such that $\cert(f|_{\rho}) \geq \bs(f|_{\rho}) = \Omega(\bs(f)^{2/3})$. That is, $f$ cannot give an improved counter-example.  
\end{theorem}
\begin{proof}
Since sensitivity condenses losslessly (see also \cref{lem:s-deg-stronger}) and $\s(f)\leq \bs(f)$, for $f$ to be a counter-example to lossless condensation of block sensitivity it must be the case that $\bs(f)=\omega(\s(f))$. In our case it means that we must have $\bs_0(f)=k_2k_3= \omega(\max\{\s_1(f),\s_0(f)\})$, where $\s_1(f)=k_1k_2$ and $\s_0(f)=2k_3$. Note it implies that $k_3=\omega(k_1)$ and $k_2=\omega(1)$.  Furthermore, there exists restrictions $\rho_1$ and $\rho_2$, each leaving $O(\bs(f))$ many variables unset, such that $\bs(f|_{\rho_1}) \geq \s_1(f) = k_1k_2$ and $\bs(f|_{\rho_2}) \geq \s_0(f) = 2k_3$. 

Also consider the restriction $\rho_3$ that leaves variables of $k_3/k_1$ copies of $g$ completely unset while setting the rest of the variables to all $0$'s.  Clearly, $|\rho_3^{-1}(\ast)|=k_3k_2$. We thus have $\bs_0(f|_{\rho_3}) = k_3k_2/k_1$. 

Now note that we would like to minimize the maximum among $\{\bs(f|_{\rho_1}), \bs(f|_{\rho_2}), \bs(f|_{\rho_3})\}$, with respect to $\bs_0(f)=k_2k_3$, to be able to achieve the maximum loss in condensation. Rewriting the three quantities in terms of $\bs_0(f)$, we have \(\bs(f|_{\rho_1})=k_1\bs_0(f)/k_3\), \(\bs(f|_{\rho_2})=2\bs_0(f)/k_2\), and \(\bs(f|_{\rho_3}\} = \bs_0(f)/k_1\).  


We now claim that the maximum in $\{ k_1\bs_0(f)/k_3, 2\bs_0(f)/k_2, \bs_0(f)/k_1\}$ is always at least $\bs_0(f)^{2/3}$. For the sake of contradiction suppose not. That is, each one of them is $< \bs_0(f)^{2/3}$. Then, from the third bound we will have $k_1 > \bs_0(f)^{1/3}$, which in turn implies $k_3>\bs_0(f)^{2/3}$ (using the first bound). This then implies, using $\bs_0(f)=k_2k_3$, that $k_2 < \bs_0(f)^{1/3}$. We then have the second bound becoming $>2\bs_0(f)^{2/3}$, and thus reaching a contradiction. 
\end{proof}

\end{document}